\documentclass[oneside, a4paper, onecolumn, 10pt]{article}

\newcommand{\thesistitle}[0]{Polynomial-time parametric optimisation}
\newcommand{\authorname}[0]{Doan Dai Nguyen}

\newcommand{\supervisor}[0]{Sarah J. Berkemer, Yann Ponty}
\newcommand{\supervisorinstitution}[0]{Laboratoire d'Informatique de l'École Polytechnique}

\newcommand{\abstracttext}[0]{%
In biology, predicting RNA secondary structures plays a vital role in determining its physical and chemical properties. Although we have powerful energy models to predict them as well as parametric analysis to understand the models themselves, the large number of parameters involved makes exploring the parameter space and effective fine-tuning complicated at best. The literature describes an approach via so-called RNA polytopes and several attempts to compute them entirely, but computing explicitly the polytopes is both practically and theoretically intractable. In this thesis, we demonstrate how to further modify the dynamic programming algorithms used in RNA secondary structure prediction, and more generally how to use only supporting functions to gather some information about the polytopes without explicit construction. We provide the mathematical frameworks with proofs or sketch thereof whenever necessary, and carry out some numerical experiments to show that our proposed methods are practical even when the number of parameters is large. As it turns out, one of our methods provides a solution to another problem in computational geometry previously unsolved to our knowledge, and we hope this thesis will accommodate future studies in RNA, as well as inspire further researches on the potential uses of polytopes' supporting functions in computational geometry.
}

\usepackage[
  left=2cm,top=2.0cm,bottom=2.0cm,right=2cm,
  headheight=17pt, 
  includehead,includefoot,
  heightrounded, 
]{geometry}

\usepackage[T1]{fontenc}
\usepackage{amstext}
\usepackage{amsmath}
\usepackage{amssymb}
\usepackage{url}
\usepackage{graphicx}
\usepackage{wrapfig}
\usepackage{enumerate}
\usepackage{paralist}
\usepackage{xspace}
\usepackage{color}
\usepackage{times}
\usepackage[colorlinks,linkcolor=blue]{hyperref}
\usepackage[colorinlistoftodos,prependcaption,textsize=normal]{todonotes}
\usepackage{pdfpages}
\usepackage{fancyhdr} 

\usepackage{titling}
\usepackage[nottoc,numbib]{tocbibind}

\usepackage{mathtools, nccmath}
\usepackage{amssymb, amsthm, mathrsfs}
\usepackage{BOONDOX-cal}

\newcommand{\conv}[1]{\text{conv}\left(#1\right)}
\newcommand{\proj}[1]{\text{proj}(#1)}
\newcommand{\vol}[1]{\text{vol}(#1)}
\newcommand{\vertex}[1]{\text{Vert}(#1)}
\newcommand{\sgn}[1]{\text{sgn}(#1)}

\newcounter{cnt}
\newcounter{def}
\newcounter{thm}
\theoremstyle{definition}
\newtheorem{definition}[def]{Definition}
\newtheorem{problem}[cnt]{Problem}
\newtheorem{theorem}[thm]{Theorem}
\newtheorem{lemma}[thm]{Lemma}

\usepackage{algpseudocode,algorithm,algorithmicx}

\algrenewcommand\algorithmicrequire{\textbf{Input:}}
\algrenewcommand\algorithmicensure{\textbf{Output:}}

\usepackage{tikz-cd}
\usepackage{booktabs,dcolumn,caption}

\newcommand\mc[1]{\multicolumn{1}{c}{#1}}
\usepackage{xcolor,colortbl}

\usepackage{graphicx}
\usepackage{subfig}
\graphicspath{ {./} }

\usepackage{pgfplots}
\pgfplotsset{compat=1.17}
\usepackage{wrapfig}

\begin{document}
\pagenumbering{gobble}
\clearpage




\hspace{0pt}
\vfill

\begin{center}

\includegraphics[width=0.3\textwidth]{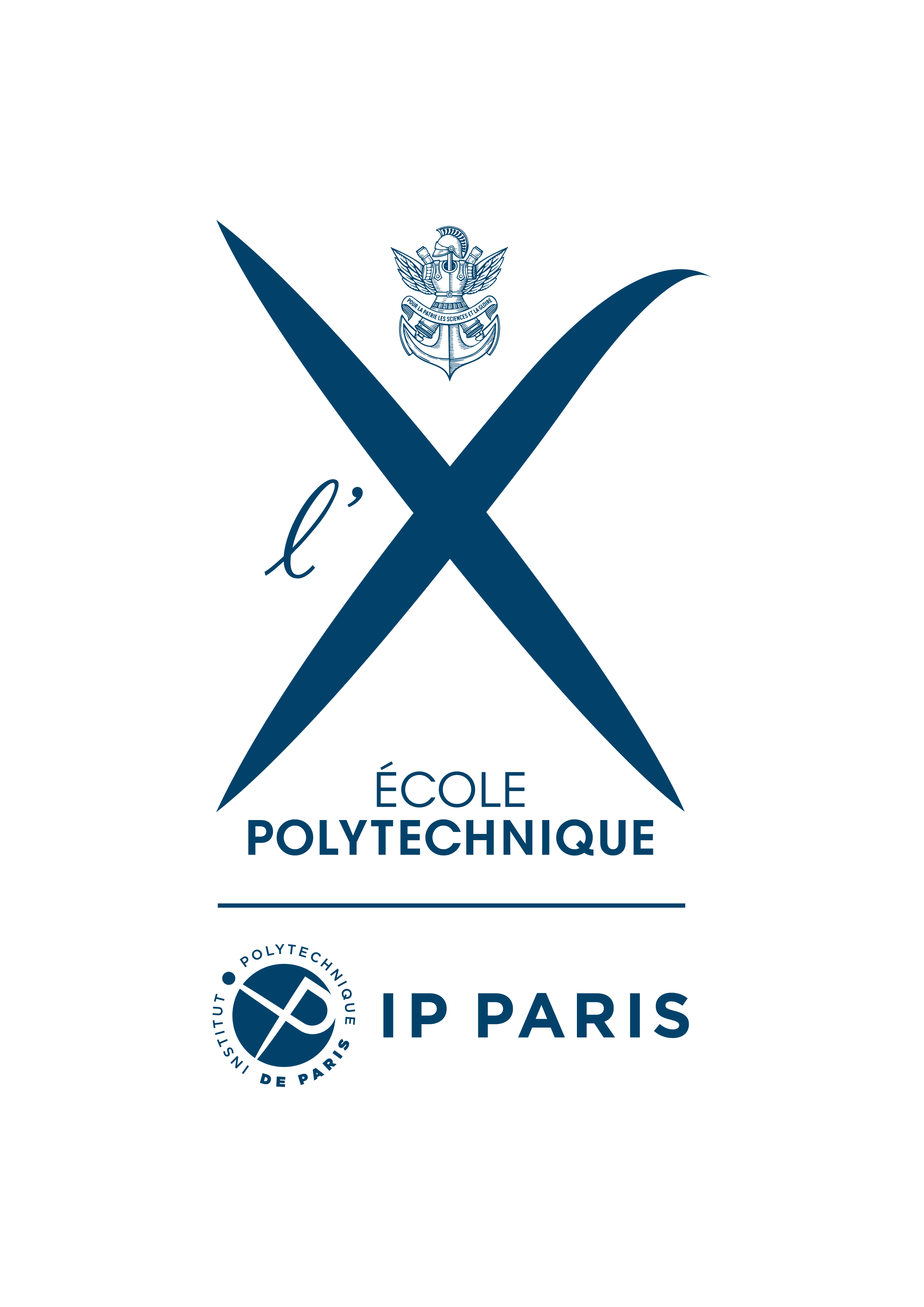}

\vspace*{2em}
{\large
\textbf{\'Ecole Polytechnique}

\vspace*{1em}
\textit{BACHELOR THESIS IN COMPUTER SCIENCE}

\vspace*{3em}
{\Huge \textbf{\thesistitle}}
\vspace*{3em}

\textit{Author:}

\vspace*{1em}
\authorname{}, \'Ecole Polytechnique

\vspace*{2em}
{\textit{Advisor:}}

\vspace*{1em}
\supervisor{}, \supervisorinstitution{}
}

\vspace*{2em}
\textit{Academic year 2022/2023}

\end{center}

\vfill
\hspace{0pt}

\newpage

\vfill
\noindent\textbf{Abstract}\\[1em]
\fbox{
\parbox{\textwidth}{
\abstracttext{}
}
}
\vfill

\newpage

\pagestyle{fancy}
\lhead{\authorname}
\rhead{\thesistitle}

\newpage
\tableofcontents
\newpage

\pagenumbering{arabic}

\section{Introduction}
Since the discovery of ribosomal RNA and transfer RNA in the late 1950s, a variety of non-coding RNA families have been discovered, which, to name a few, includes long non-coding RNAs \cite{Statello2021}, circular RNAs \cite{Liu2022}, and antisense RNAs and CRISP RNAs\cite{Ashwath2023}, each with potential for new therapies. Non-coding RNAs have shown to act as catalysts for chemical reactions, as gene regulators \cite{Statello2021}, and as DNA replicator \cite{Kuhnlein2021}, thereby are believed to play a vital role in the origin of life. We now know that at least 76\% of human genome is transcribed, yet only 1.2\% of which encodes proteins \cite{Lee2019}, and most of the rest, i.e. non-coding RNAs, have their functions yet to be completed unveiled, and applications to be discovered. Thus, there raises the need to understand properties and functions of RNAs.

On the one hand, advances in sequencing techniques have resulted in an exponential growth for the total length of sequences for the past decade \cite{Katz2022}. On the other hand, inferring structure of RNA by experimental methods, despite much progress in recent years, remains both arduous, time-consuming, and technically demanding. There have been attempts to combine experiment data with simulation, yet this approach has yet to overcome the limitation in both accuracy and throughput \cite{Liu2022}. Although RNA has its tertiary structure largely determined by its secondary structure \cite{Tinoco1990}, the same challenges persist, making prediction by computational methods being most viable approach, amongst which is free energy minimisation method.

Whilst delaying formal definition to later sections, we recall that a RNA sequence consists of nucleotides, commonly denoted by four letters $A$, $C$, $G$, and $U$. A free energy model of RNA then admits, amongst others, the following assumptions:
\begin{itemize}
	\item An RNA exists as it is and folds as a whole, contrasting with co-transcription model where RNA folds as it is being transcribed.
	\item An RNA admits only the canonical base pairs $A-U$ and $G-C$, and the so-called wobble pairs $G-U$.
	\item An RNA folds in isolation, independent from the surrounding environment. 
\end{itemize}

A general thermodynamic framework consists of a model which assigns to a given pair of sequence and secondary structure a measure, e.g. energy, entropy, and/or enthalpy, with respect to the state where no nucleotides are paired. Then, under some assumptions such as pseudoknot-free and heuristics regarding the energy function, it is possible to decompose the total energy as those of substructures, whereby a dynamic programming algorithm can exploit the sub-problem hierarchy to minimise or maximise this energy, before tracing back to find an optimal solution efficiently. One may see that whilst the performance bottleneck may lie in the algorithm, the limit of prediction capability lies in the underlying energy model, whence comes the need for a biologically accurate energy model.

At the one end of the spectrum is the model counting number of base pairs as featured in work of Nussinov and Jacobson \cite{Nussinov1980a}, whose simplicity allows full analysis of the parameter space but severely limits the prediction capacity. At the other end for pseudoknot-free secondary structures is Turner energy model \cite{Mathews1999a} with over 7600 features, generally considered to be biologically realistic, but at the same time difficult to well-tune. A subset of parameters were measured by optical melting experiments, but a large part was derived by fitting to the experiment data.

Despite its comprehensiveness and even regarding pseudoknot-free RNA secondary structures, Turner model fails to predict accurately in many cases. Additionally, the ad-hoc energy function for multi-loops, originally derived for simplicity and algorithm derivation's sake, outperforms other more complicated and realistic alternatives \cite{Ward2017}. These two phenomena beg the question if it is due to a suboptimal choice of parameters or the fundamental limit of the model itself. This line of study is often called parametric analysis, where one explores the parameters to observe what a model can predict and argue about its capability, and is of particular importance, for this model is also applied to the study of single-stranded DNA \cite{SantaLucia2004}, which has both an essential role in virology \cite{Malathi2019} and potential in therapeutics \cite{Avci-Adali2015}. 

Part I (Section \ref{chapter: notation and preliminaries}-\ref{chapter: current methods part 1}) focuses on reviewing methods in parametric analysis, where Section \ref{chapter: parametric analysis and polytopes part 1} commences with original motivation and gives a particular example in the context of RNA secondary structures. We then elaborate and derive such an approach using polytopes, based on works of Pachter and Sturmfels \cite{Pachter2004, Pachter2005}, before conclude on the two necessary operations from computational geometry and their intractability.

In Part II (Section \ref{chapter: introduction part 2}-\ref{chapter: ellipsoid method part 2}), we focus on models with large number of features, and provide three new methods for studying two properties of the polytope corresponding to an RNA sequence and an energy model, namely robustness and learnability, the latter of which is related to collision detection problem and is called \emph{Relative position problem} throughout this thesis. Section \ref{chapter: introduction part 2} formally defines the two problems, reviews current methods including the so-called MPR algorithm, and shows how we need further studies for them to be applicable to our interest. Section \ref{chapter: conclusion} concludes this thesis with outlines of future works, including applications beyond the scope of energy model in Bioinformatics and parametric analysis in computational geometry.

Aside from what was presented above, the main contribution of this thesis is as follows:
\begin{enumerate}
	\item We present a derivation of polytope algebra simpler than that by Pachter and Sturmfels \cite{Pachter2004, Pachter2005} (Section \ref{chapter: parametric analysis and polytopes part 1}).
	\item We review the current methods in computation geometry, which leads to the conclusion that current algorithms are computationally incapable of carrying out parametric analysis for large number of features (Section \ref{chapter: current methods part 1}).
	\item We present a new scheme to modify the underlying dynamic programming algorithms to obtain robustness with little overhead in runtime and space (Section \ref{chapter: telescoping method part 2}).
	\item Concerning relative position problem, we lay the foundation for a new family of methods, which includes MPR algorithm but also other variations, amongst which we give examples of three such variants. Numerical experiments show they have comparable and altogether remarkable performances in practice, especially considering that we have no proofs of terminating, and even supposing that such a method halts, we also demonstrate a poor worst-case performance (Section \ref{chapter: simplex method part 2}).
	\item Finally, we incorporate work of Hornus \cite{Hornus2017} with ideas from ellipsoid method in linear programming to develop a novel approach which has performance expected to be poor in practice, but nonetheless demonstrates that this problem is solvable in weak polynomial time (Section \ref{chapter: ellipsoid method part 2}). To our best knowledge, this is the first algorithm to do so.
\end{enumerate}

\subsection*{Acknowledgement}

The author wishes to thank Sarah J. Berkemer and Yann Ponty for the topic suggestion and countless helpful discussions. This research is funded by Laboratoire d'Informatique de l'\'Ecole Polytechnique.
\section{Notation and preliminaries}
\label{chapter: notation and preliminaries}

In what follows, let $X = \mathbb{R}^{d+1}$ or $\mathbb{S}^d$ for $d > 1$. We shall consider the Euclidean space $\mathbb{R}^{d+1}$ equipped with the canonical inner product $\langle \cdot, \cdot \rangle$. For $m \leq d$, we define a $m$-plane to be a $m$-dimensional subspace of $\mathbb{R}^{d+1}$. Unless explicitly used for other purposes, uppercases letters represent polytopes, polyhedra, and hyperplanes, whilst lowercase letters represent affine points as column vectors, thus for $x, y \in \mathbb{R}^{d+1}$, we have $\langle x, y \rangle = x^T y$. We denote $[x, y]_{\mathbb{R}^n}$ as the line segment uniquely given by two points $x, y \in \mathbb{R}^{d+1}$, where the subscript shall be neglected whenever the context is apparent.

Except what is defined in this thesis and with possibly different notations, the definitions used throughout this thesis may be found in Ziegler's \emph{Lectures on Polytopes} \cite{Ziegler1995} and in Ratcliffe's \emph{Foundations of Hyperbolic Manifolds} \cite{Ratcliffe2007}.

\subsection{Spherical geometry}
Let $k \leq d+1$, an intersection of $\mathbb{S}^d$ and an $k$-plane is called a great $k$-sphere, which for terminology consistency, we shall call a $k$-plane whenever the context is clear. The case $k = 2$ gives great circles dividing $\mathbb{S}^d$ into halfspheres. And, just as any line segment is the shortest path connecting the two endpoints, any arc of a great circle is a geodesic on $\mathbb{S}^d$, which gives the notion of distance between two points $x, y \in \mathbb{S}^d$ as the length of the arc connecting them, which is necessarily unique (even if such an arc is not).  Geometrically, this is equal to the angle between the two vector $x$ and $y$, giving the formula $d_{\mathbb{S}^d}(x, y) = \cos^{-1} \langle x, y \rangle$.

Given this metric, given two points $x, y \in \mathbb{S}^d$ and $d \geq d_{\mathbb{S}^d}(x, y)$, we can define a spherical ellipsoid $\mathcal{E}$ on $\mathbb{S}^d$ as the set of points $z$ such that 
\[
d_{\mathbb{S}^d} (x, z) + d_{\mathbb{S}^d} (z, y) \leq d.
\]
Equivalently, $\mathcal{E}$ can be viewed as the intersection of $\mathbb{S}^d$ and an elliptic hypercone.

For any two distinct non-antipodal points $x, y \in \mathbb{S}^d$ which we shall call a proper pair of points, there exists a unique geodesic segment defined by the arc connecting them, which we denote as $[x, y]_{\mathbb{S}^d}$. The subscript shall be neglected whenever the context is apparent.

\subsection{Convex set}
A set $K \subseteq X$ is convex if for any two points $x, y \in K$, one has $[x, y] \in K$. It is clear that the intersection of two convex sets are convex, and $X$ is convex, hence any set $S \subseteq X$ admits a minimal convex set containing it, called the convex hull of $S$ and denoted as $\conv{S}$.
\[
\conv{S} = \bigcap \left\{K' \mid S \subseteq K' \subseteq X, K' \text{convex}\right\}.
\]

Now we consider $X = \mathbb{R}^{d+1}$, $K \subseteq \mathbb{R}^{d+1}$ compact and convex, and for $y \in \mathbb{R}^{d+1}$, we define the supporting function of $K$ at $y$ as
\[
h_K(y) = \sup_{x \in K} \langle x, y \rangle, 
\]
where one can show the maximum is attained on the boundary of $K$. For computational purpose, we define an extremal function $\sigma_K$ as an oracle returning such a vector $x \in \partial K$ at which the minimum is attained, i.e. $\langle \sigma_K(y), y \rangle = h_K(y)$, chosen arbitrarily if there exists many.
Let $A$ and $B$ be two convex sets, we define Minkowski sum of $A$ and $B$ as $A + B = \{a + b \mid (a, b) \in A \times B\}$. It can be proven that $A + B$ is convex, and if $A$ and $B$ are compact, then so is $A + B$.

\subsection{Polytopes}
For a finite set of points $P = \{x_1, x_2, ..., x_n\} \subseteq X$, its convex hull $\conv{P}$ is a convex polytope, whose dimension is defined to be that of the minimal plane (i.e. great spheres in case $X = \mathbb{S}^d$) containing $P$. Conversely, any convex polytope can be represented as the convex hull of such a set $P$, giving the so-called $\mathcal{V}$-representation. By abuse of notation, we shall refer to $P$ as the polytope $\conv{P}$ whenever the context is apparent. It is clear that $\conv{P}$ is compact, and in what follows, we only consider convex polytopes, thus omit the word "convex" when speaking of polytopes. 

A face $F$ of $P$ is said to be $k$-dimensional if it is contained in a minimal $k$-plane. The faces of $P$ are themselves polytopes, with dimension ranging from 0, the \emph{vertices}, to 1, the \emph{edges}, and up to $d$, the \emph{facets}, and $d+1$, which is $P$. Two faces are said to be adjacent if their intersection is non-empty. We denote $\vertex{P}$ to be the set of vertices of $P$.

In case $X = \mathbb{R}^{d+1}$, $P$ admits a normal fan $\mathcal{N}(P)$ associating to each face $F$ of $P$ the set of vector $y$ such that $h_P(y)$ is attained by only and any point $x \in F$, which can be shown to be a cone. For each cone $N \in \mathcal{N}(P)$, we define the normal spherical polytope $S = \mathbb{S}^d \cap N$. By convexity of $P$, $\mathcal{N}(P)$ is a complete fan, so the collection of such polytopes, denoted $\mathcal{S}(P)$, covers $\mathbb{S}^d$.

By abuse of notation, we denote $h_{\conv{P}}$ and $\sigma_{\conv{P}}$ by $h_P$ and $\sigma_P$ whenever they are well-defined, respectively. Let $A$ and $B$ be two polytopes, then one can show that $A + B$ is also a polytope.

By Minkowski-Weyl theorem, a polytope can be represented either as a set of vertices, also known as $\mathcal{V}$-representation, or as a bounded intersection of some halfspaces, which in $\mathbb{R}^{d+1}$ amounts to specifying a normal vector for each facet, also known as $\mathcal{H}$-representation. Thus, we shall call the complexity of a polytope $P$ to be the total number of its vertices and facets.

Finally, we define formally the robustness of a vector, whose name shall be justified in the next section.
\begin{definition}
	Let $P$ be a polytope, $p \in \mathbb{R}^{d+1}$, we define the robustness of $p$ (with respect to $P$), denoted $\theta_P(p)$, to be the supremum of angle $\theta$ such that for all $p' \in \mathbb{R}$ satisfying $\cos^{-1} ((p, p')) = \frac{\langle p, p'\rangle}{\| p \| \| p' \|} \leq \theta$ and for all $x \in P$, if $\langle x, p \rangle = h_P(p)$ then $\langle x, q \rangle = h_P(q)$, or equivalently, if $\langle x, p \rangle < h_P(p)$ then $\langle x, q \rangle < h_P(q)$. 
\end{definition}
Since if $\langle x, p \rangle = h_P(p)$ then for any $\lambda \geq 0$, one has $\langle x, \lambda p \rangle = h_P(\lambda p)$, it is thus customary to consider $p \in \mathbb{S}^d$, and the notion of robustness coincides with the minimum distance from $p$ to another normal spherical polytope $S \in \mathcal{S}(P)$ not containing $p$, which since there are only finitely many polytopes all of which are closed, the supremum is attained.

\subsection{RNA sequence and secondary structure}

A prominent example throughout this thesis is RNA secondary structure prediction problem, for which this reminder may prove to be useful. An RNA sequence $q = q_1 q_2 ... q_n$ is a string of nucleotides, each of which is one of the four letters (bases) $A$, $C$, $G$, and $U$. Then, given such a sequence, a base pair can form between two distinct positions $i < j$, and denoted by the unordered pair $\{i, j\}$. Hereinafter, we shall admit only the canonical base pairs and the wobble pairs, i.e. one must have $\{s_i, s_j\} = \{A, U\}$, $\{G, C\}$, or $\{G, U\}$. However, each base can belonged to at most one base pairs.

A set of $k$ such unordered pairs $s = \{\{i_\ell, j_\ell\} \mid \ell=1,2,...k\}$ is called a secondary structure, and said to be compatible if the base pairing are all admitted. Two base pairs $\{i, j\}$ and $\{k, l\}$ such that $i < k < j < l$ or $k < i < l < j$ are said to be crossing, and thus form what is called a pseudo-knot, and shall not be concerned in this thesis. Given an RNA sequence $q$ and a compatible secondary structure $s$, this assumption allows one to decompose a secondary structure into various smaller structures of some types, and thus may use dynamic programming to calculate certain features of the pair $(q, s)$. For instance, we may count the number $n_{ij}$ of base pairs of type $i-j$, which can be summarised by a triplet $c(q, s) = (c_{AU}, c_{GC}, c_{Gu}) \in \mathbb{N}^3$, which we call a feature vector or a signature.

Continue with this example, we may impose an energy model over these features, by choosing a parameter set $p = (p_{AU}, p_{GC}, p_{GU}) \in \mathbb{R}^3$. Then, for a pair $(q, s)$, we have the energy to be $E = \langle c(q, s), p \rangle = c_{AU} p_{AU} + c_{GC} p_{GC} + c_{GU} p_{GU}$. And moreover, given $q$ and $p$, one can find a compatible structure $s$ minimising $E$ by Nussinov's algorithm \cite{Nussinov1980a}: in particular, let $f(i, j)$ be the minimum energy over the subsequence $q' = q_i q_{i+1} ... q_{j-1} q_j$, one has
\[
	f(i, j) = \max
	\begin{cases}
		f(i+1, j) \\
		f(i, j-1) \\
		f(i-1, j+1) + p_{\{q_i, q_j\}} \text{ if } q_i \text{ and } q_j \text{ can form a base pair} \\
		\max_{i < k < j} f(i, k) + f(k+1, j)
	\end{cases},
\]

\begin{wrapfigure}[13]{l}{0.47\textwidth}
	\vspace{-12pt}
	\centering
	\begin{tabular}{c|c|c|c|c|c|c|c|c|c|c|}
		\mc{}& \mc{U} & \mc{A} & \mc{U} & \mc{U} & \mc{C} & \mc{U} & \mc{G} & \mc{A} & \mc{U} & \mc{G} \\
		\cline{2-11}
		U & 0 & 1 & 1 & 1 & 1 & 2 & 2 & 3 & 4 & 4 \\
		\cline{2-11}
		A & \cellcolor{gray} & 0 & 0 & 0 & 0 & 1 & 2 & 2 & 3 & 3 \\
		\cline{2-11}
		U & \cellcolor{gray} & \cellcolor{gray} & 0 & 0 & 0 & 1 & 2 & 2 & 3 & 3 \\
		\cline{2-11}
		U & \cellcolor{gray} & \cellcolor{gray} & \cellcolor{gray} & 0 & 0 & 1 & 2 & 2 & 3 & 3 \\
		\cline{2-11}
		C & \cellcolor{gray} & \cellcolor{gray} & \cellcolor{gray} & \cellcolor{gray} & 0 & 1 & 2 & 2 & 3 & 3 \\
		\cline{2-11}
		U & \cellcolor{gray} & \cellcolor{gray} & \cellcolor{gray} & \cellcolor{gray} & \cellcolor{gray} & 0 & 0 & 1 & 1 & 1 \\
		\cline{2-11}
		H & \cellcolor{gray} & \cellcolor{gray} & \cellcolor{gray} & \cellcolor{gray} & \cellcolor{gray} & \cellcolor{gray} & 0 & 1 & 1 & 1 \\
		\cline{2-11}
		A & \cellcolor{gray} & \cellcolor{gray} & \cellcolor{gray} & \cellcolor{gray} & \cellcolor{gray} & \cellcolor{gray} &  \cellcolor{gray} & 0 & 1 &  1\\
		\cline{2-11}
		U & \cellcolor{gray} & \cellcolor{gray} & \cellcolor{gray} & \cellcolor{gray} & \cellcolor{gray} & \cellcolor{gray} &  \cellcolor{gray} & \cellcolor{gray} & 0 & 0 \\
		\cline{2-11}
		G & \cellcolor{gray} & \cellcolor{gray} & \cellcolor{gray} & \cellcolor{gray} & \cellcolor{gray} & \cellcolor{gray} &  \cellcolor{gray} & \cellcolor{gray} & \cellcolor{gray} & 0\\
		\cline{2-11}
	\end{tabular}%
	\vspace{-5pt}
	\captionsetup{justification=centering}
	\captionof{table}{Memoisation table for $q = UAUUCUGAUG$. \label{tab: example of memoisation table}}
\end{wrapfigure}

together with initialisation $f(i, i) = f(i, i-1) = 0$ for all $i$. Then, given the memoisation table $f$, one can trace back to find an optimal and compatible secondary structure $s$.

Note that this formulation of Nussinov's algorithm has inherent ambiguity, and in particular, a secondary structure can be traced back in different ways, but this does not change the final output and thus of our concern. Likewise, one signature can correspond with multiple secondary structures, and depending on the choice of parameters, it is possible that multiple signatures are optimal.

For instance, let $p = (1, 1, 1)$ and $q = UAUUCUGAUG$, we have the memoisation table presented in Table \ref{tab: example of memoisation table}. We present in Figure \ref{fig: examples of optimal structures} some of the optimal secondary structures. Note that some base pairs are between adjacent nucleotides, which are not biologically realistic. Such base pairs are forbidden in Turner model, and can also be incorporated into Nussinov's algorithm, but we decide to omit in this example for the sake of simplicity.

\begin{figure}[H]%
	\centering
	\subfloat[An optimal structure with signature $(2, 1, 1)^T$]{{\includegraphics[width=0.45\textwidth]{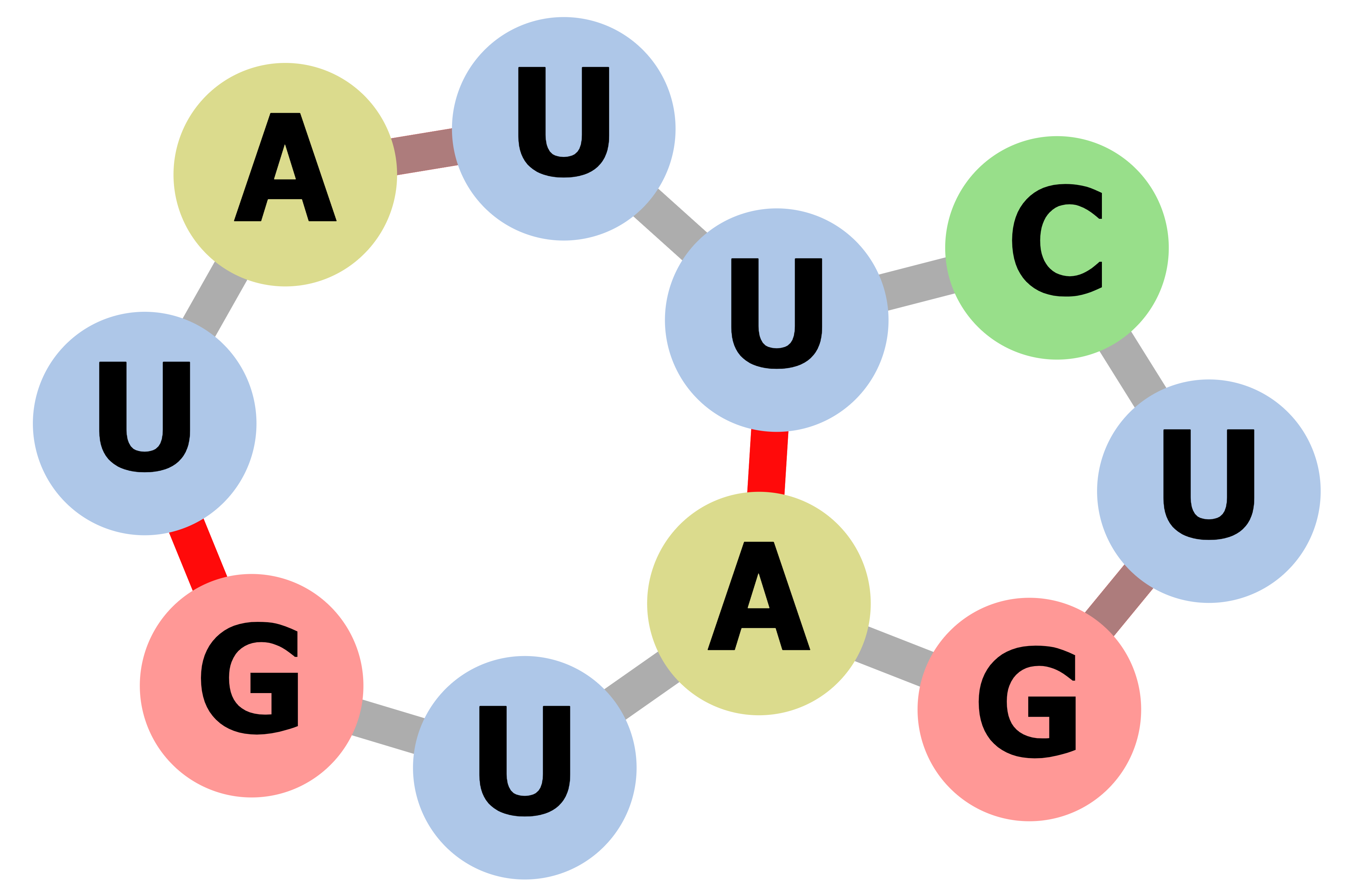} }}%
	\qquad
	\subfloat[An alternative optimal structure with signature $(2, 0, 2)^T$]{{\includegraphics[width=0.45\textwidth]{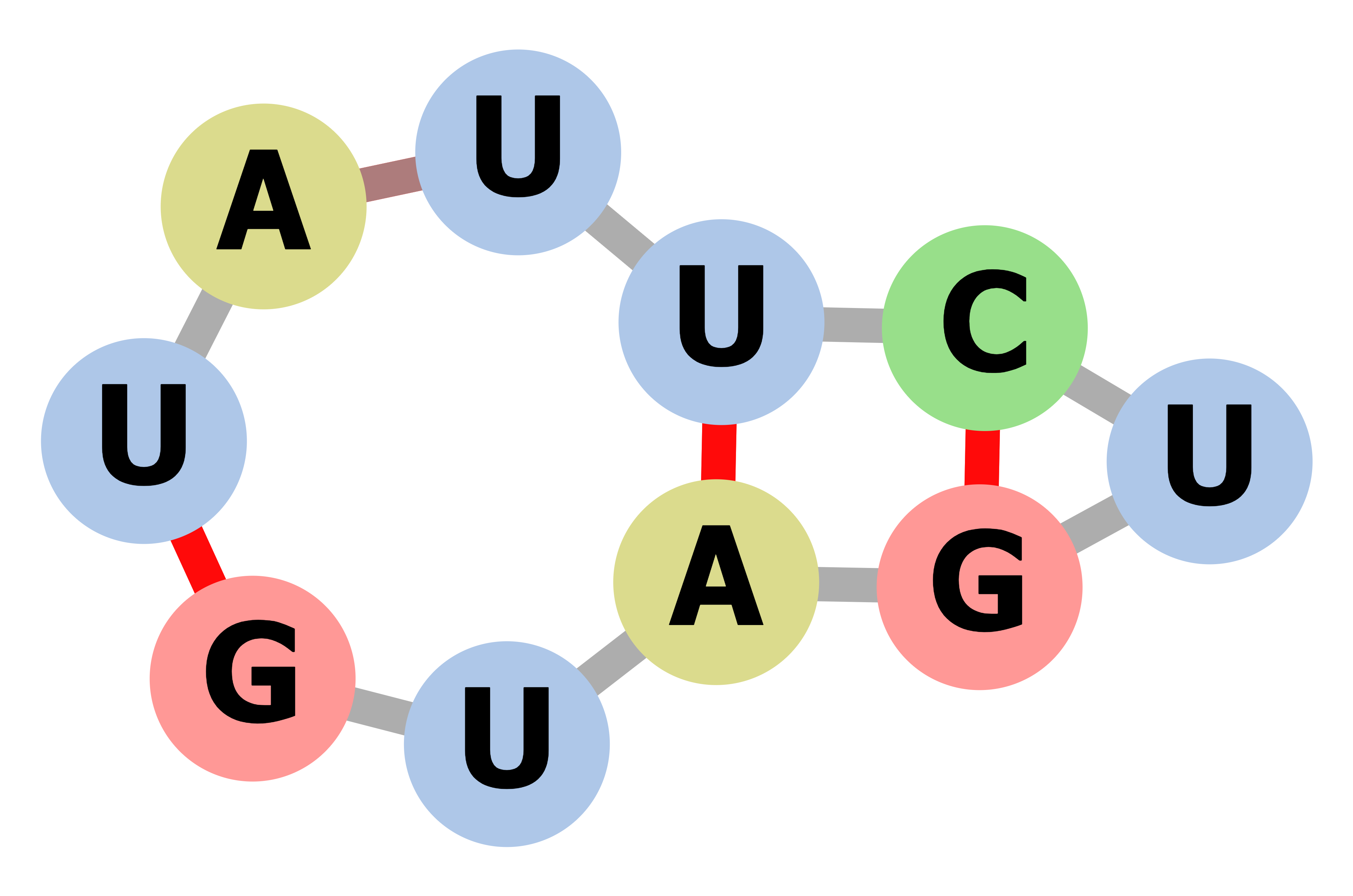} }}%
	\\
	\subfloat[An alternative optimal structure with signature $(2, 1, 1)^T$...]{{\includegraphics[width=0.45\textwidth]{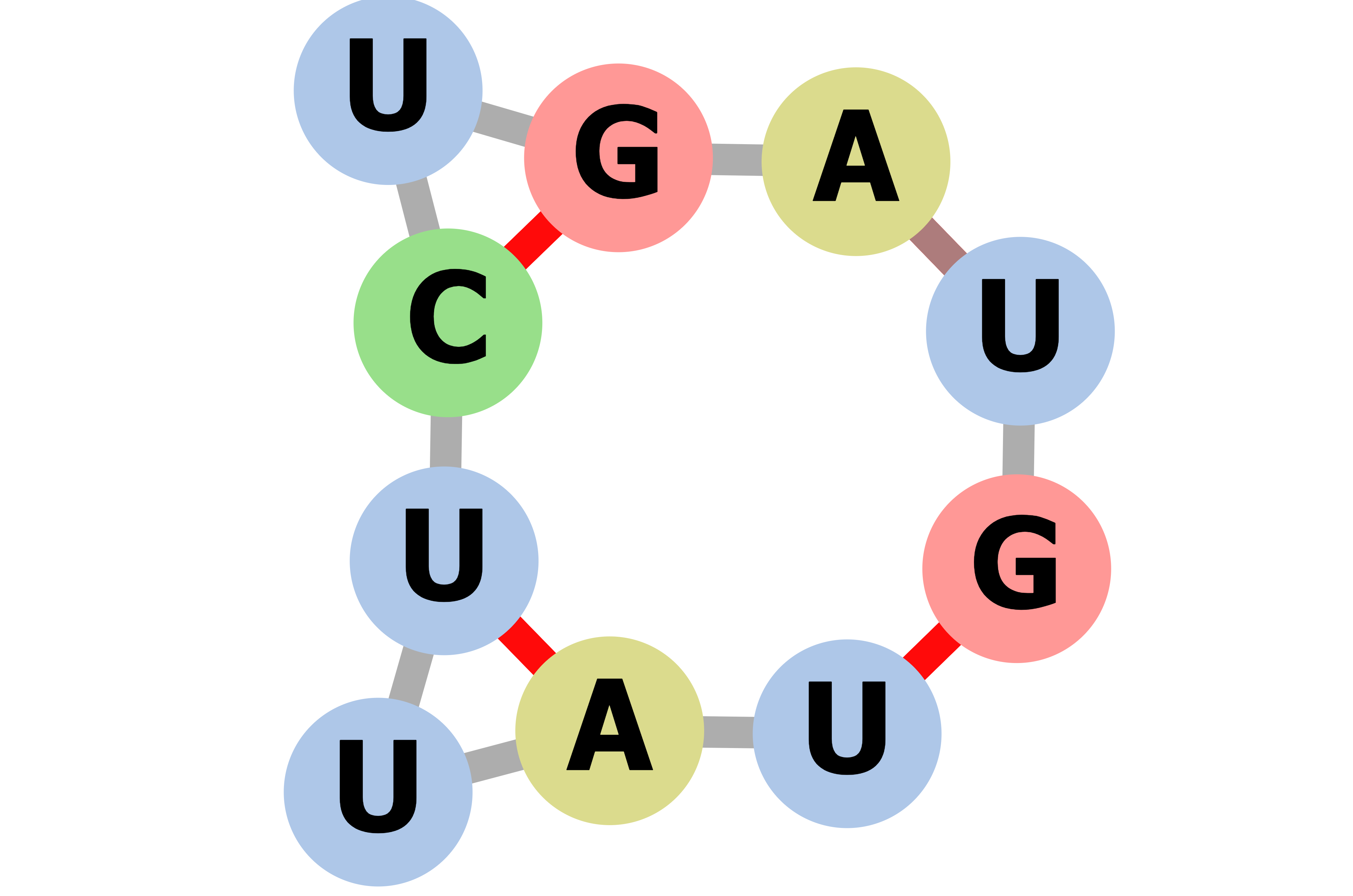} }}%
	\qquad
	\subfloat[... and another]{{\includegraphics[width=0.45\textwidth]{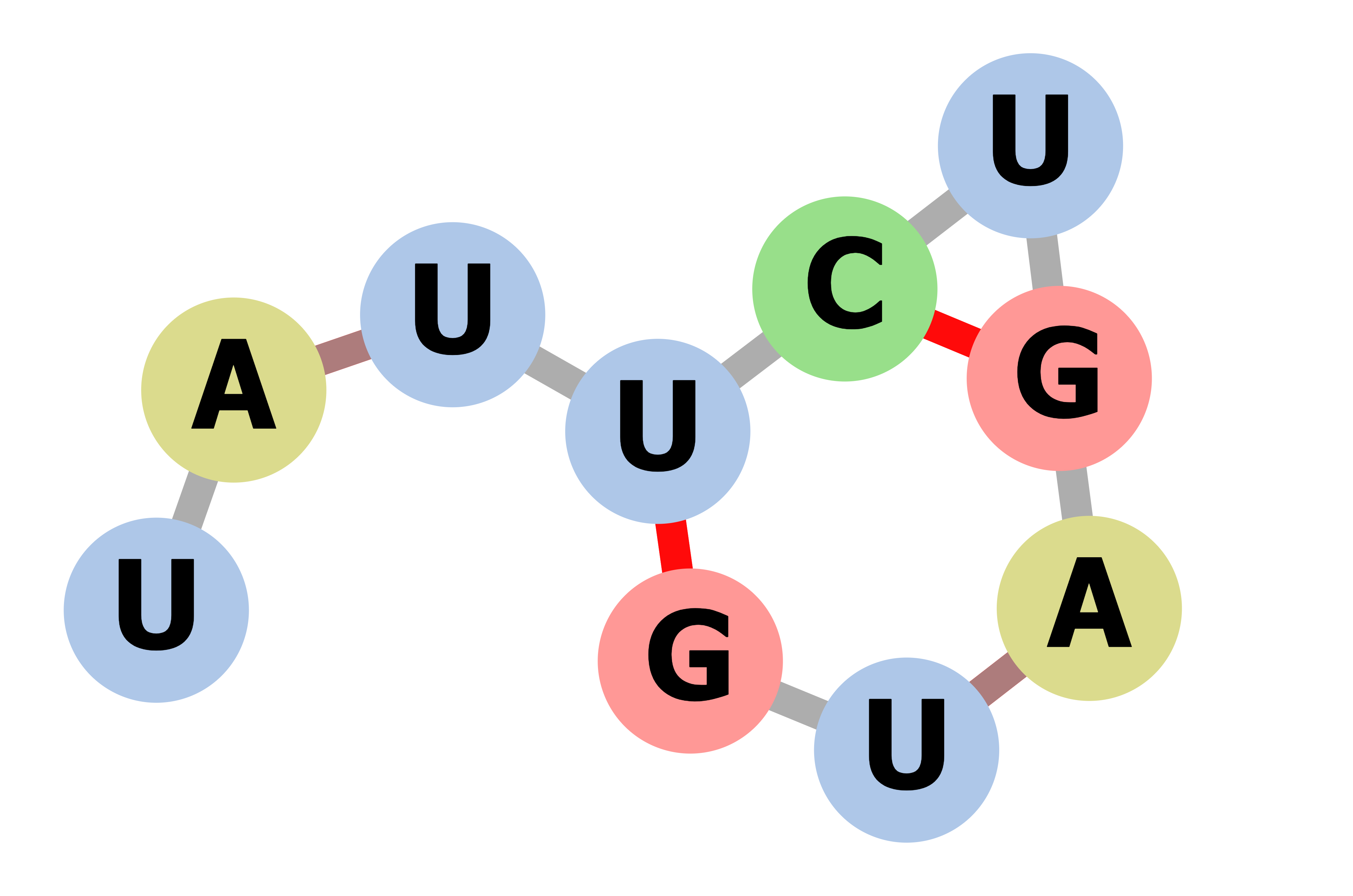} }}%
	\\
	\subfloat[Similarly, an alternative optimal structure with signature $(2, 0, 2)^T$...]{{\includegraphics[width=0.45\textwidth]{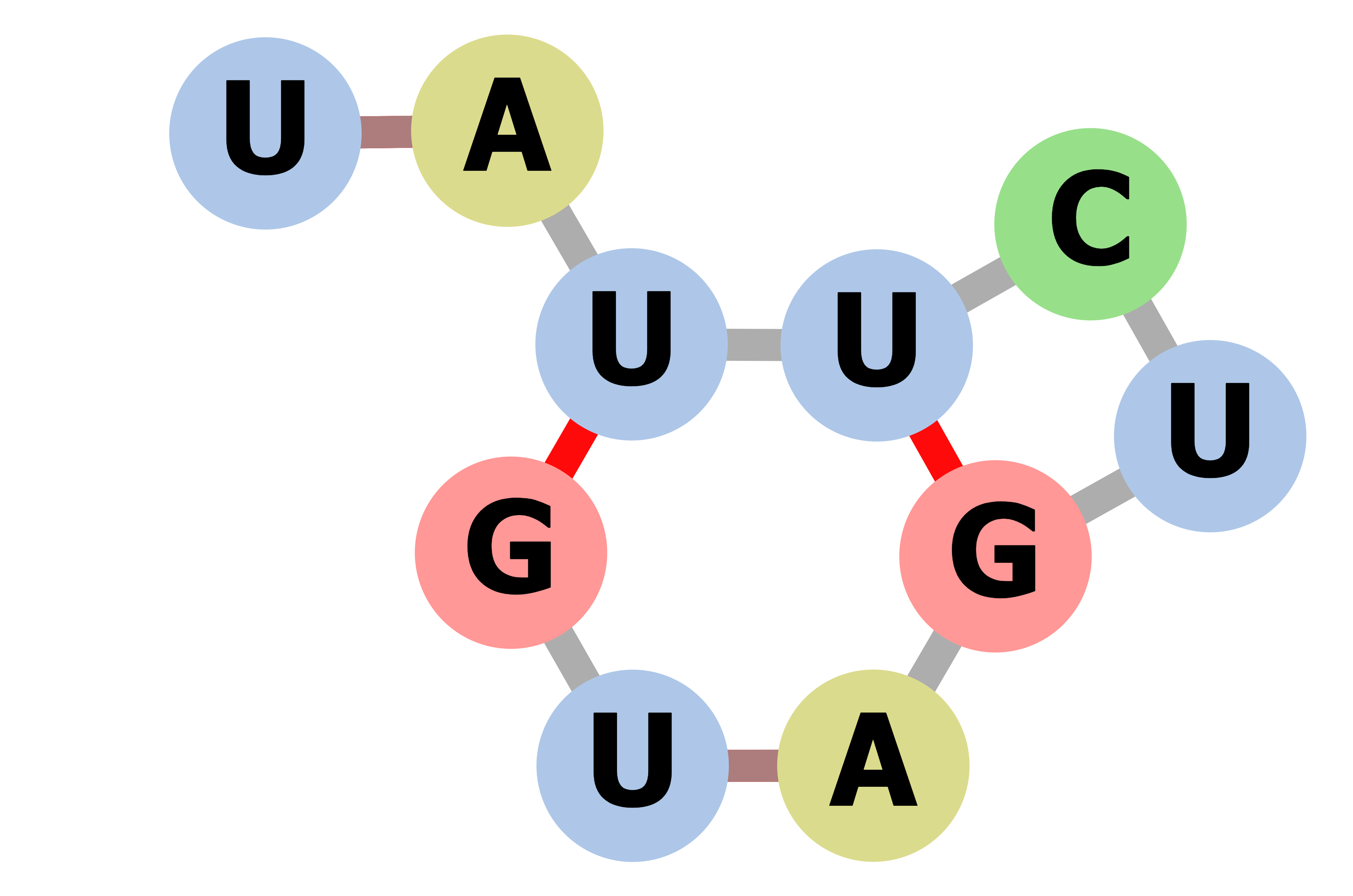} }}%
	\qquad
	\subfloat[... and yet another]{{\includegraphics[width=0.45\textwidth]{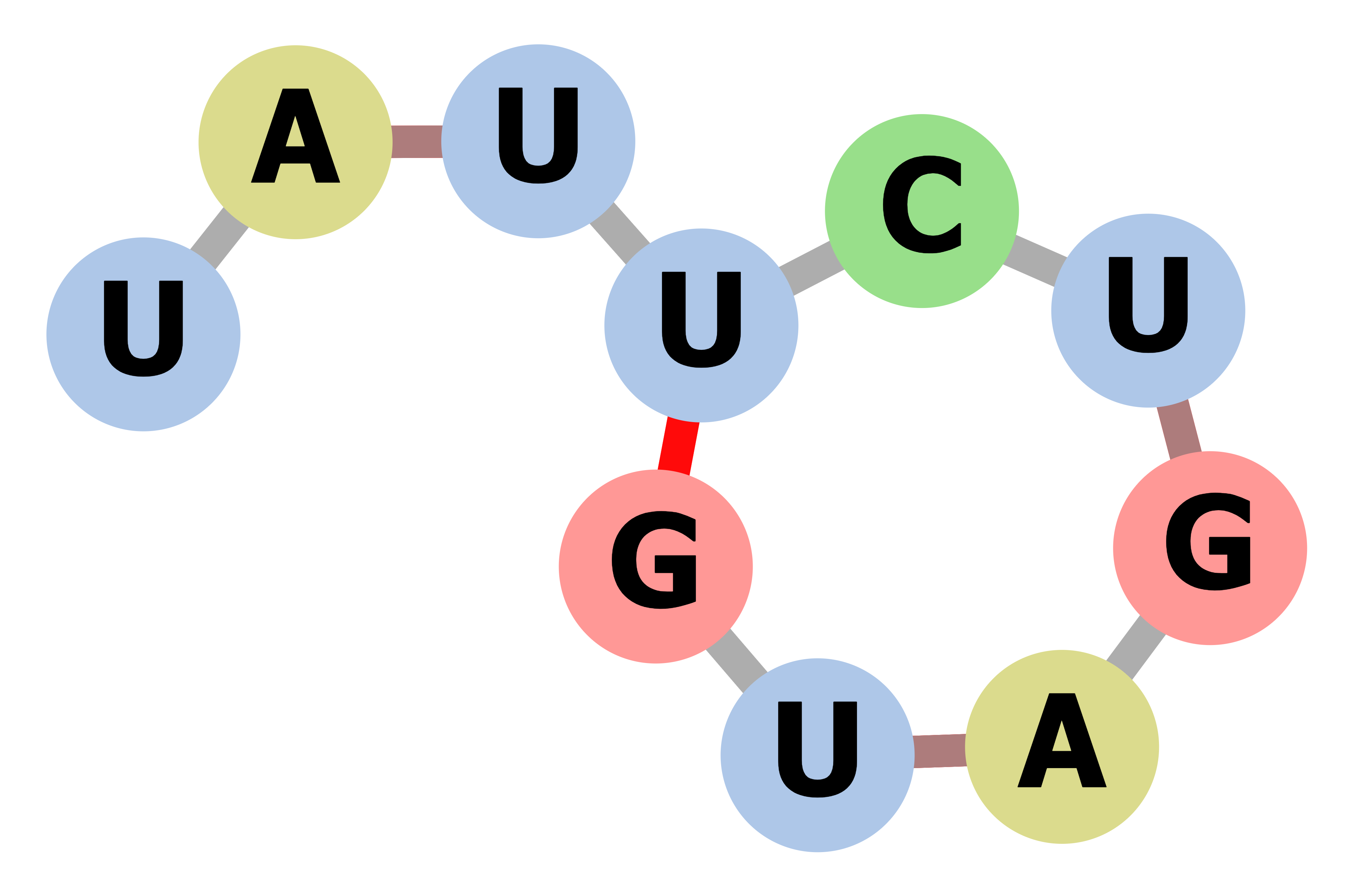} }}%
	\caption{Examples of optimal structures for $q = UAUUCUGAUG$. Non-gray edges denote base pairs.}%
	\label{fig: examples of optimal structures}
\end{figure}
\section{Parametric analysis and polytopes}
\label{chapter: parametric analysis and polytopes part 1}

\subsection{Original motivation}
The parametric analysis problem originally arose from sequence alignment algorithms where it was unclear how to specify the parameters. Thus, one wishes to study what parameter sets will give rise to what alignment, i.e. a way to classify parameter sets. In one of the earliest attempts, Fitch and Smith \cite{Fitch1983} considered a 2-parameter function to measure similarity in Needleman-Wunsch algorithm, and studied two short sequences derived from mRNA of chicken $\alpha$- and $\beta$-hemoglobin, where they identified 11 possible optimal solutions. Their method involved computing alignment for a number of parameter sets to identify a region in which any parameter set would yield the same alignment. They identified the region and moved to its neighbours, sequentially searched through the parameter space until no such region could be found. Nonetheless, this approach requires heavy computation, redundant alignments, and overall poses difficult as to argue about the regions, e.g. proving that no other regions can exist.

One crucial observation is by the discrete nature of alignments (or, in case of RNA, secondary structures), there exist necessarily finitely many possible optimal solutions, even if one considers \emph{all} uncountably many possible parameter sets. This number is further reduced since we consider not the alignments (resp. secondary structures) themselves, but some features whereof.

In our example, we focus only on possible values of $c_{AU}$, $c_{GC}$, and $c_{GU}$. A brute-force approach reveals that for $q = UAUUCUGAUG$, there are 67 compatible secondary structures, whence arise 15 possible signatures, thus there can be no more than 15 possible optimal structures even if we consider all parameter sets. These signatures are given below.

\begin{quote}
	\centering
	$(0, 0, 0)$, $(0, 0, 1)$, $(0, 1, 0)$, $(0, 1, 1)$, $(0, 0, 2)$,
	
	$(1, 0, 0)$, $(1, 0, 1)$, $(1, 1, 0)$, $(1, 1, 1)$, $(1, 0, 2)$,
	
	$(2, 0, 0)$, $(2, 0, 1)$, $(2, 1, 0)$, $(2, 1, 1)$, $(2, 0, 2)$
\end{quote}

Unfortunately, this approach is not generalisable as the length of $q$ grows: intuitively, a longer sequence is expected to have more possible structures, and if nucleotides are uniformly distributed, then this number can grow exponentially. Indeed, Zukor and Sankoff \cite{ZUKER1984} demonstrated the following theorem.

\begin{theorem} (Zuker and Sankoff \cite{ZUKER1984}, 1984)
	Let $q$ be a sequence of length $n$, whose bases are given by i.i.d random variables with the probability of occurrence for $A$, $G$, $C$, $U$ to be $a$, $g$, $c$, and $u$ respectively. Denote $p = 2(au + gc)$, $\alpha = \left(\frac{1 + \sqrt{1+4\sqrt{p}}}{2}\right)^2$, and $H = \frac{\alpha(1+4\sqrt{p})^{\frac{1}{4}}}{2\sqrt{\pi}p^\frac{3}{4}}$, then as $n$ tends to infinity, one has the expected number of compatible secondary structures $E(n)$ to be
	\[
		E(n) \sim Hn^{-\frac{3}{2}}\alpha^n.
	\]
	In particular, $p = \frac{1}{4}$ for the case where all nucleotides can occur with equal probability gives $\alpha = \frac{1+\sqrt{3}}{2} = 1.866...$.
\end{theorem}

Much less is known about the number of possible signatures, as it is necessarily model-dependent. For our model counting the number of base pairs for each type and in general for any energy model, it is expected that the small number of features will greatly reduces the possible signatures, for the value of each feature is bounded by the length $|q|$ of $q$. Therefore, assuming the number of features is fixed, the number of possible signatures will be bounded by a polynomial of $|q|$. But we remind that it is one thing to compute \emph{the number of signatures}, it is another thing to compute \emph{the signatures} themselves: even in our model where signatures necessarily have integer coordinates, a priori there is no viable way to know which signatures our model can predict.

\subsection{Reduction to polytopes}

The next crucial observation is we need not care about all signatures, but only those whom our model can predict. Mathematically, suppose a parameter $p$, a model $E$ given by a set of features given by a function $c(\cdot, \cdot)$, and a sequence $q$, a dynamic programming algorithm will then compute
\[
	\max_{s'} \langle c(q, s'), p \rangle,
\]
where $s'$ ranges over all compatible secondary structures. This motivates us to define the RNA polytope of $q$ (in the model $E$) as
\[
	\mathcal{P}(q) = \conv{\{c(q, s) \mid \text{secondary structure } s\}},
\]
then a structure $s$ that our model can predict necessarily satisfies $\langle c(q, s), p \rangle = h_{\mathcal{P}(q)}(p)$, or equivalently $c(q, s) \in \partial \mathcal{P}(q)$. We call such a pair $(q, s)$ \emph{learnable} (for model $E$), and we may restrict ourselves to studying only the boundary of $P$.

\begin{figure}[h]%
	\centering
	\subfloat[\centering RNA polytope $\mathcal{P}(q)$]{{\includegraphics[width=0.5\textwidth]{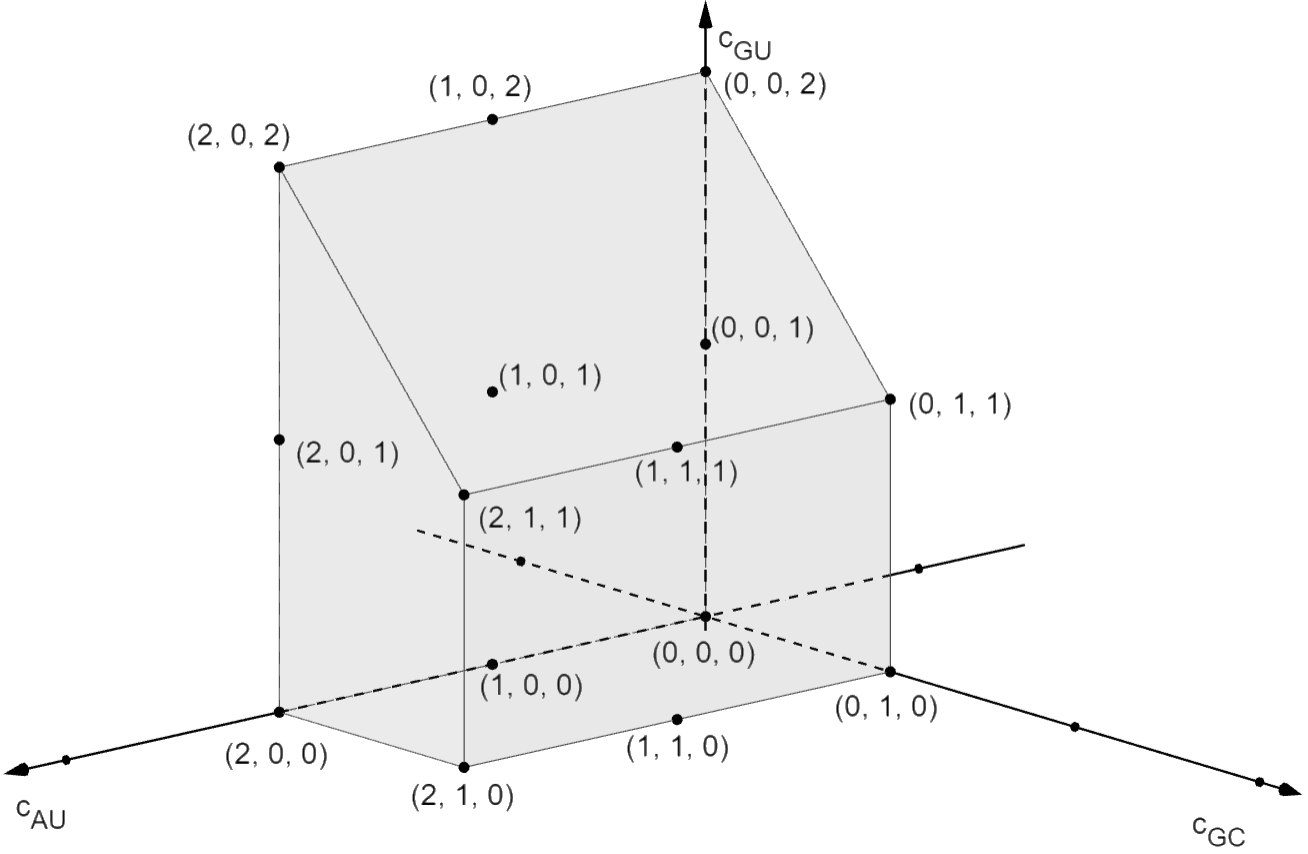} }}%
	\qquad
	\subfloat[\centering spherical normal polytopes $\mathcal{S}(\mathcal{P}(q))$]{{\includegraphics[width=0.4\textwidth]{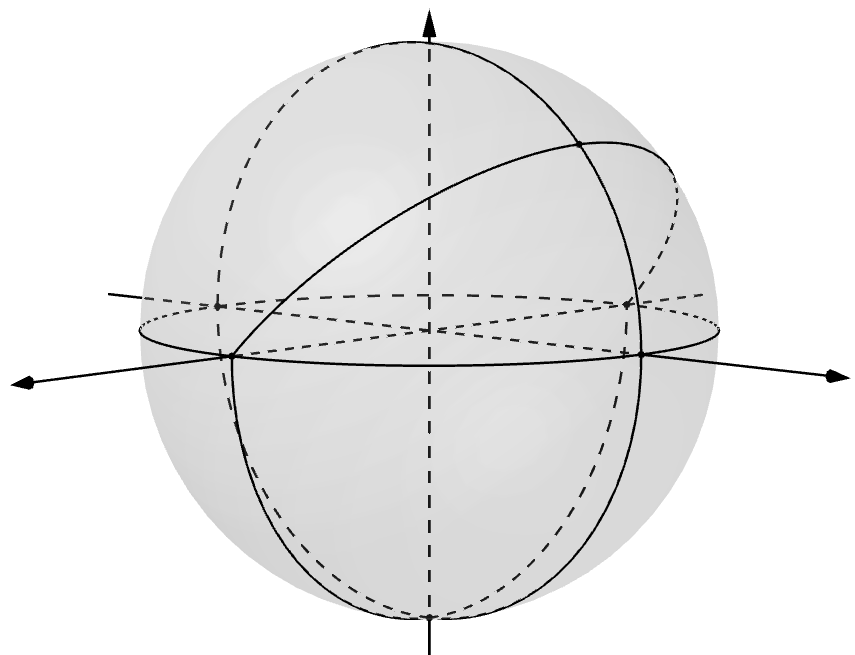} }}%
	\caption{Example of an RNA polytope for $q = UAUUCUGAUG$}%
	\label{fig:example of RNA polytope}%
\end{figure}

Back to our example, from the list of 15 possible signatures, we have the polytope $\mathcal{P}(q)$ and its spherical normal polytopes $\mathcal{S}(\mathcal{P}(q))$ shown in Figure \ref{fig:example of RNA polytope}. 

We have three remarks:
\begin{enumerate}
	\item Given a parameter set $p$, finding the signature it will predict corresponds to finding the spherical polytope it belongs to, thus this gives us a complete classification of parameter sets.
	\item In our example, it happens to be the case that all signatures lie on the boundary of $\mathcal{P}(q)$, but as the length increases, this phenomenon is not to be expected in general. Whilst we found no results concerning the complexity of $\mathcal{P}(q)$ as it depends not only on the number of possible signatures, but also \emph{their distribution} \cite{Har-Peled2011}, which has not been well-studied for any model. Regarding the sequent alignment problem, amongst $O(n^d)$ signatures, where $n$ and $d$ denote the total length of two sequences and the dimension, Gusfield et al. \cite{Gusfield1994} showed that only $O(n^{d\frac{d-1}{d+1}})$ points lie on the boundary for $d = 2$, and Pachter and Sturmfels \cite{Pachter2004} showed the same for general $d$.
	\item Finally, this phenomenon of unlearnability is not restricted to high dimension, and in fact it is usually the opposite: for a given sequence, lower dimensions allow fewer signatures to be learnable. Section \ref{chapter: current methods part 1} in particular shows how it can happen for $d = 2$.
\end{enumerate}

In the early 1990s, various methods were proposed to construct systematically the decomposition of parameter space with less computations, such as that by Fernandez-Baca and Srinivasan \cite{Fernandez-Baca1991} (amongst others \cite{Vingron1994}). Their algorithm involves finding an initial convex set of points in the polytope $P$, then gradually extending the set until it matches the boundary of $P$. This scheme is applicable to general dimension, with complexity to be polynomial of that of $P$ and $\sigma_P$.

Nonetheless, the methods until then treated $\sigma_P$ and $h_P$ as black boxes, which both introduced unnecessary overheat and was not entirely satisfied in the context of bioinformatics in general, where many of the algorithms are dynamic-programming-based. Examples include Needleman-Wunsch and Smith-Waterson algorithm for sequence alignment, Nussinov's and Zuker's algorithm for RNA secondary structure, and Fitch's algorithm for phylogenetic tree construction. Since the optimal solution is constructed by solving subproblems, one may ask if it is possible to construct $\mathcal{P}(p)$ in a similar fashion by modifying the original dynamic programming scheme. One then may hope to find a more efficient mean to construct the polytopes, and even gather information about it without explicit constructions.

\subsection{Dynamic programming in polytope algebra}

In the context of graphical models, Pachter and Sturmfels \cite{Pachter2004, Pachter2005} demonstrated that an alternative approach to construct the polytope from the dynamic programming scheme is possible. In particular, they showed that for any models whose optimal value is a linear combination of the parameters which can be found by a max-sum decomposition, there is a natural extension of the dynamic programming to compute the corresponding polytope of the model. In particular, recall the (max) tropical semi-ring given by $(\mathbb{R}\cup\{-\infty\}, \max(\cdot, \cdot), +)$, one can consider the same max-sum decomposition but in polytope algebra, given by $(\{\text{convex polytopes}\}, \oplus, \otimes)$, where for any two given polytopes $A$ and $B$, one defines

\[
	A \oplus B = \conv{\{A \cup B\}}, A \otimes B = A + B.
\]

\begin{figure}
	\centering
	\subfloat[\centering $A$]{{\includegraphics[width=0.2\textwidth]{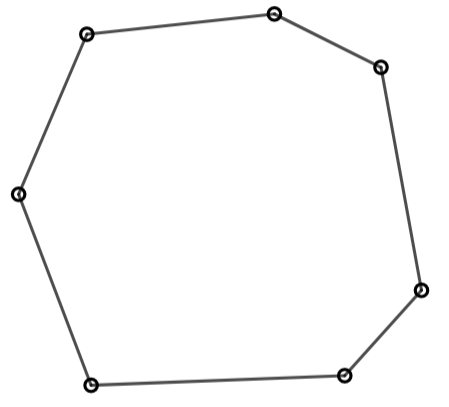} }}%
	\qquad
	\subfloat[\centering $B$]{{\includegraphics[width=0.2\textwidth]{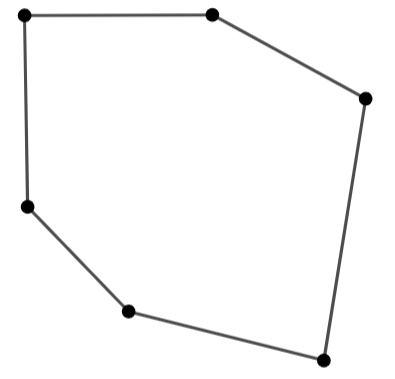} }}%
	\qquad
	\subfloat[\centering $A \oplus B$]{{\includegraphics[width=0.2\textwidth]{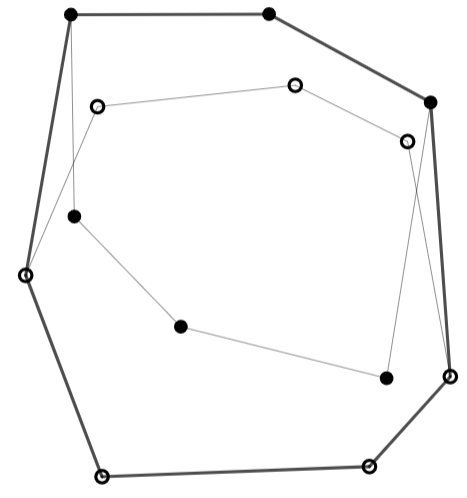} }}%
	\qquad
	\subfloat[\centering $A \otimes B$]{{\includegraphics[width=0.2\textwidth]{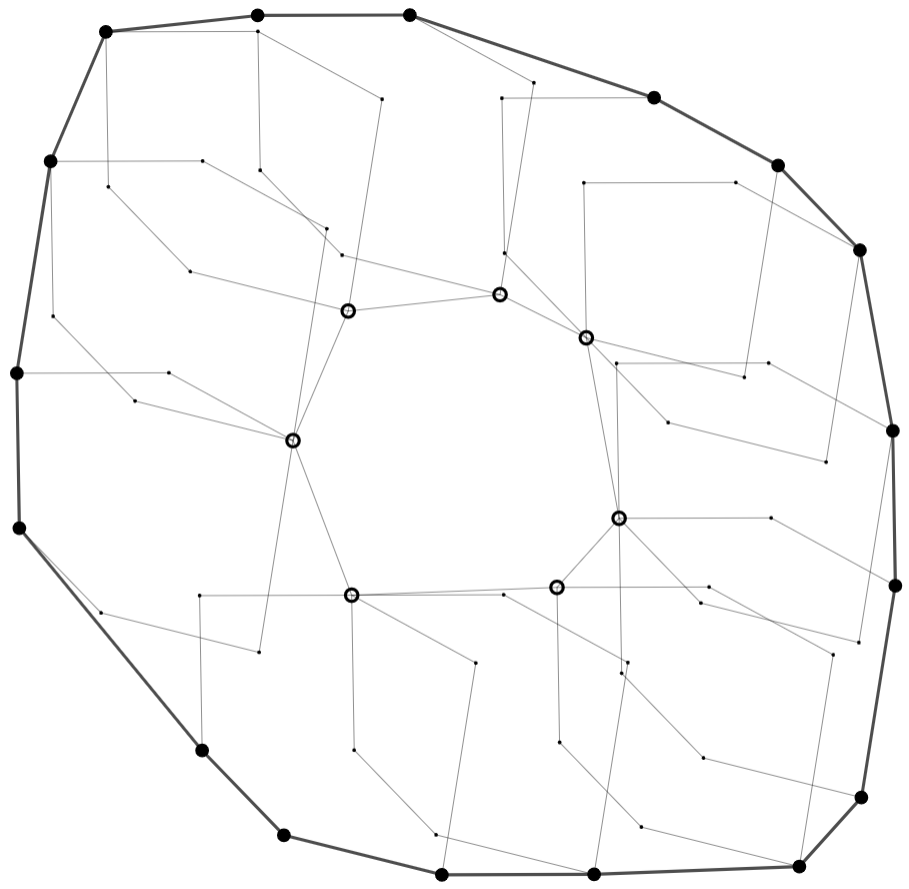} }}%
	\caption{Example of operations in polytope algebra.}
\end{figure}

A more throughout treatment can be found in original articles, but for sake of simplicity, in this thesis, we shall derive the scheme in a more intuitive fashion, and without machinery used by Pachter and Sturmfels, such as Newton polytopes.

First, we note that the Nussinov's algorithm in our example uses only two operations, namely addition and taking the max. It turns out to be the paradigm for many other dynamic programming algorithm, that Tendeau formalised in definition \cite{Tendeau1998}, and thus in this case, instead of computing in full the dynamic programming to trace out the boundary of $\mathcal{P}(q)$, we only need to keep track of how the polytopes evolve as the algorithm is executed.

Given two terms $a(p)$ and $b(p)$ (which stands for values of form $f(i, j)$ for some $i$ and $j$ in our cases) whose exact values depend on the choice of parameter set $p$ at the beginning of execution, to which we associate two polytopes $A$ and $B$. Recalling the definition of such polytopes, we have that $a(p) = h_A(p)$, and likewise, $b(p) = h_B(p)$. Therefore, suppose we associate to the term $\max(a(p), b(p))$ a polytope $C$, then one must have $h_C(p) = \max(a(p), b(p)) = \max(h_A(p), h_B(p))$ for all $p$; similarly, to the term $a(p) + b(p)$ the corresponding polytope $D$ must satisfy $h_D(p) = a(p) + b(p) = h_A(p) + h_B(p)$.

On the other hand, consider the polytope $C' = A \oplus B$ and $D' = A \otimes B$, and some $p \in \mathbb{R}^{d+1}$, then by construction, one has
\[
	h_{C'}(p) = \max_{x \in C'} \langle x, p \rangle = \max_{x \in A \cup B} \langle x, p \rangle = \max\left(\max_{x \in A} \langle x, p \rangle, \max_{x \in B} \langle x, p \rangle \right) = \max(h_A(p), h_B(p)),
\]
and similarly, 
\[
	h_{D'}(p) = \max_{x \in D'} \langle x, p \rangle = \max_{(a, b) \in A \times B} \langle a + b, p \rangle = \max_{(a, b) \in A \times B} \langle a, p \rangle + \langle b, p \rangle = \max_{a \in A} \langle x, p \rangle + \max_{b \in B} \langle b, p \rangle,
\]
so we conclude that $h_{C'} = h_C$ and $h_{D'} = h_D$. As a convex set is determined by its supporting function, we conclude that $C = A \oplus B$ and $D = A \otimes B$.

With a mathematical foundation laid down, we can at least in theory carry out the modified dynamic programming to study the parametric space. For instance, going back to our example, we can modify the Nussinov's algorithm to polytope algebra, yielding the following recurrent relation, where $P_{AU} = (1, 0, 0)^T$, $P_{GC} = (0, 1, 0)^T$, $P_{GU} = (0, 0, 1)^T$, and $F(i, j)$ denotes the polytope associated with the $q_i q_{i+1} ... q_{j-1} q_j$.

\[
F(i, j) = \bigoplus
\begin{cases}
	F(i+1, j) \\
	F(i, j-1) \\
	F(i-1, j+1) \otimes P_{\{q_i, q_j\}} \text{ if } q_i \text{ and } q_j \text{ can form a base pair} \\
	\bigoplus_{i < k < j} F(i, k) \otimes F(k+1, j)
\end{cases}.
\]
Similar to original Nussinov's algorithm mentioned in Section \ref{chapter: notation and preliminaries}, this recurrent equation has redundancy, which can be overcome by alternative formulations. But this shall not change the final output, nor alter the complexity up to a constant factor.
\section{Current methods}
\label{chapter: current methods part 1}

Unfortunately, supporting both operations of the polytope algebra poses a real challenge: whilst DNA sequence alignment and base-pair counting model each have 3 features, Turner energy model has close to 8000. Suppose we can effectively construct the feature vector from a given pair of sequence and secondary structure, the large dimensionality will induce a bottleneck. For general dimension $d$, there are currently two approaches.

\subsection{Computation in full dimension}
One may attempt to compute $\mathcal{P}(q)$ in full, yet the challenge lies on how one represent the polytopes. Let $A$ and $B$ be two polytopes, then $A \oplus B$ and $A \otimes B$ are efficiently computable only when $A$ and $B$ are given in $\mathcal{H}$- and in $\mathcal{V}$-representation, respectively \cite{Tiwary2008}.

We may attempt to maintain both representations: this is known as the Double Description method, whose dual version is known as Beneath-and-Beyond algorithm, but this approach is also not computationally feasible \cite{Bremner1999}. In practice, even in small dimensions, this approach still causes problems as the length of $q$ increases: when we restrict ourselves to the study of multi-loops in Turner model, which corresponds to 3 features, Poznanovic et al. showed that computing $\mathcal{P}(q)$ when $q$ is a tRNA of length 50 nts took approximately 2 hours, and when $q$ is a 5S rRNA, it took on average 23 hours. They also reported that a sequence of length 175 nts would increased the computation time to a week, and for $q$ of length 354 nts would take more than 2 months. Thus, it cannot practically cover the effective length for which Turner model finds its applications, which is up to 700 nts \cite{Mathews1999a}, and not scalable for the data we have.

It is natural that as the length of $q$ increases, so will the complexity of $\mathcal{P}(q)$. On the one hand theory, the result above by Zuker and Sankoff provides some intuition. On the other hand, much less can be said about the asymptotic complexity of $\mathcal{P}(q)$, as it does not solely depend on the number of points $c(q, s)$, but also the \emph{distribution} of such points \cite{Har-Peled2011}, which has not been well-studied. In practice, Poznanovic et al. demonstrated that an increase of less than 50 nts multiplies that number by a factor of 3.5. 

To confirm this state of affairs, we implement the modified Nussinov's algorithm naively in Python, where the two operations are supported by \texttt{SageMath} software system. A random sequence of 10 nts took approximately 0.5 seconds, but that of 100 nts took 20 minutes, and increasing the length by 50 nts extended the runtime to 1 hour, agreeing with observations by Poznanovic et al. that the bottleneck lies in the algorithmic aspect, i.e. the inherent difficulty of the computational geometry problem involved, and not of sequence. 

We may also compute both operations in the same representation, but Tiwary \cite{Tiwary2008} showed that this problem is output-sensitive \emph{strongly} \texttt{NP}-hard. In other words, he showed that unless $\texttt{P} = \texttt{NP}$, there exists no algorithm to this problem whose complexity is polynomial with respect to the final polytope's complexity (which we recall to possibly be exponential of $q$'s length) and the maximum absolute value of coefficient.

Another idea is that if we represent polytopes as sets of vertices, as the dynamic programming algorithm executes, $A$ and $B$ are not sets of random points, and thus the employment of a general-purpose convex hull algorithm to compute $A \oplus B$ may not be necessary. In particular, we only need to \emph{merge} two polytopes, for which there exist such a method in case of 2- and 3-dimensional \cite{ORourke1998}, summarised in Table \ref{tab:comparison of complexity}.

\begin{table}[h]
	\centering
	\begin{tabular}{c c c}
		\toprule
		Dimension & \multicolumn{2}{c}{Complexity}       \\
		\cmidrule{2-3}
		& Merge 2 convex hulls    & Build from 2 convex hulls   \\
		\midrule
		2         & $O(n)$                  & $O(n \log h)$ \\
		3         & $O(n)$                  & $O(n \log h)$ \\
		\bottomrule\\
	\end{tabular}\\
	\begin{minipage}{300pt}
		\caption{Best complexity for merging two convex hulls versus computing from scratch. $n$ and $h$ denote number of points in input and of the convex hull, respectively. \label{tab:comparison of complexity}}
	\end{minipage}
\end{table}

We remark that $n$ and $h$ in Table \ref{tab:comparison of complexity} concern the input and output of \emph{individual} operations in polytope algebra, and not those of the derived dynamic programming algorithms. In particular, for some energy models, it may be the case that the final polytope has few vertices, yet its construction is still time-consuming. And unfortunately, it is difficult to implement already in 3-dimensional \cite{ORourke1998} and unclear how to generalise for higher dimensions, thus renders this approach intractable.

\subsection{Dimensionality reduction: an useful heuristics of uncertain precision}

A common approach is to consider a few features whilst fixing the remaining parameters, as explored by Poznanovic et al. \cite{Poznanovic2021}. Unfortunately, we lose degrees of freedom in the process, which may make a pair $(q, s)$ learnable for the full model, but not for the restricted version. Or in other words, reducing dimension allows fewer signature to be learnable.

In our example earlier, suppose we ignore the feature $c_{GC}$, then we have only 9 signatures left, namely $(0, 0)$, $(0, 1)$, $(0, 2)$, $(1, 0)$, $(1, 1)$, $(1, 2)$, $(2, 0)$, $(2, 1)$, and $(2, 2)$. Drawing on the plane, one obtains a square grid, as shown in Figure \ref{fig:example of RNA polytope reduction}. On the other hand, the signature $(1, 1)$ lies in the interior of the square and thus is not learnable.

\begin{figure}[h]%
	\centering
	\subfloat[\centering RNA polytope in 3D $\mathcal{P}(q)$]{{\includegraphics[width=0.55\textwidth]{rna_polytope.png} }}%
	\qquad
	\subfloat[\centering Its projection onto $c_{AU} - c_{GU}$ plane $\mathcal{S}(\mathcal{P}(q))$]{{\includegraphics[width=0.35\textwidth]{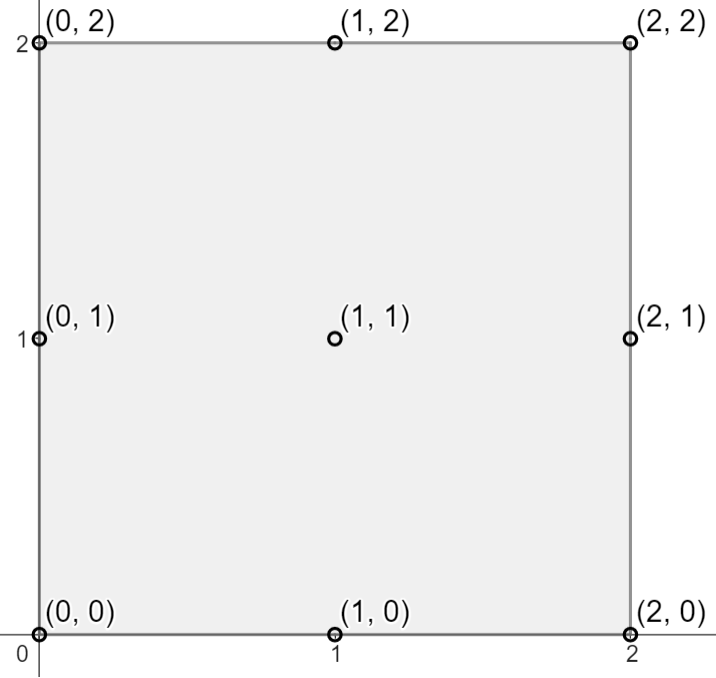} }}%
	\caption{Example of a reduced RNA polytope for $q = UAUUCUGAUG$}%
	\label{fig:example of RNA polytope reduction}%
\end{figure}

To see how this happens, let $d > 0$, $E$ be an energy model with a parameter vector $p = (p_1, p_2, ..., p_d)^T \in \mathbb{R}^d$, which assigns to each pair of sequence $q$ and secondary structure $s$ the energy $E(p, q, s) = \langle c(q, s), p \rangle$. We denote $\text{proj} : \mathbb{R}^{d+1} \mapsto \mathbb{R}^d$ the map projecting vectors of $\mathbb{R}^{d+1}$ onto $\mathbb{R}^d$ by omitting the last coordinate. Let $E'$ be a model whose energy function is $E'(p, q, s) = \langle \proj{c(q, s)}, \proj{p} \rangle$, $q$ be a sequence, $s$ be some secondary structure, $x = c(q, s)$, $x' = \proj{x}$.

Assuming we fix $p_d \neq 0$ and suppose, without loss of generality, that $p_d > 0$. Since if there exists a vector $r$ such that $E(r, q, s) = \max_{s'} E(r, q, s')$, then for any scalar $\lambda \geq 0$, one has $E(\lambda r, q, s) = \max_{s'} E(\lambda r, q, s')$, in fact one need not constraint the choices of parameters for $E$ to those of the form $\begin{pmatrix}
	p' \\
	p_d
\end{pmatrix}$ for $p' \in \mathbb{R}^{d-1}$, but in fact, one may choose any vector $r \in \mathbb{R}^d$ such that $r_d > 0$ and $E(r, q, s) = \max_{s'} E(r, q, s')$, and with $\lambda = \frac{p_d}{r_d}$, one obtains a satisfying parameter vector $p = \lambda r$. Unfortunately, there exists no such other vector $r$, for if $r_d = 0$ will lead to $p_d = 0$ for any choice of $\lambda$, and if $r_d < 0$, then choosing $\lambda = \frac{p_d}{r_d}$ results in $E(p, q, s) = \min_{s'} E(p, q, s')$. Thus, if $(q, s)$ is learnable for $E$ but only with parameter vectors $r$ such that $r_d < 0$, then $(q, s)$ is not learnable for $E'$, no matter how one varies other parameters.

In rough term, for any non-zero parameter fixed, we lose a halfsphere of $\mathcal{S}(\mathcal{P})(q)$, so if we fix $k$ parameters, there remains only $2^{-k}$ part of the parameter space that we can explore. In case of Turner model where $k$ is more than 7500, one can see how this affects the performance. Unfortunately, this case is of our greatest interest, since specifying a parameter to be zero effectively ignores the corresponding feature entirely and does not bring any new information.

As to \emph{how much} we can lose, assuming that $\mathcal{P}(q)$ is the realisation of a random variable, Amenta and Ziegler \cite{Amenta1996} showed that we are guaranteed to retain at least a certain portion of learnable structures regardless of $q$'s length. However, little is known about the precise bound.
\section{Learnability and Robustness: parametric analysis in high dimensions}
\label{chapter: introduction part 2}

In the last part, we have re-derived the polytope approach to the parametric analysis problem, and showed how it could be use to determine, given a pair of RNA sequence $q$ and a secondary structure $s$, whether the pair $(q, s)$ is learnable for a given energy model $E$, meaning if there exists a parameter set $p$ for which $(q, s)$ minimises the energy $\langle c(q, s), p \rangle$.

We have also reviewed the two current approaches to the computational aspect, namely computing the full polytopes and/or considering its projection to a lower dimension. We showed that each approach has its own limits, but to conclude, suppose we wish to study Turner energy model, even if we combine both approach and consider only 3 features, as shown by Poznanovic et al. \cite{Poznanovic2021} as well as by our experiments, the parametric analysis is still practically intractable.

Nonetheless, for some problems, it is not necessary to construct the whole polytope, especially when we concern only some properties of $\mathcal{P}(q)$ relevant to our pair $(q, s)$. Given a RNA sequence $q$, instead of finding what secondary structures (or, to be more precise, those of what signatures) an energy model $E$ can predict, we can ask the following question:
\begin{quote}
	Suppose we observe a structure $s$ for a sequence $q$ in experiments, does there exist a parameter set $p$ of which $E$ can predict $s$?
\end{quote}
Mathematically speaking, we ask if there exists a parameter vector $p$ such that $\langle c(q, s), p \rangle = h_{\mathcal{P}(q)}(p)$, or equivalently, if $c(q, s)$ lies on the boundary of $\mathcal{P}(q)$, i.e. if $(q, s)$ is learnable for $E$. We call this the learnability problem, and note that although there is a difference amongst predicting the structure $s$, the signature $c(q, s)$, and the energy $\langle c(q, s), p \rangle$. But, if we can find a parameter $p$ for which $s$ is an optimal solution, Wuchty et al. \cite{Wuchty1999} demonstrated an algorithm to recover \emph{all} secondary structures with energy close to $\langle c(q, s), p \rangle$. This plays a great role, considering that the majority of experimentally observed secondary structures have energy close to the minimum, and that in real world, a RNA may exhibit multiple suboptimal structures (cf. Chapter \ref{chapter: conclusion}).

To generalise, we have a fast access to $\mathcal{P}(q)$ via the dynamic programming algorithm, which takes a parameter set $p$ and returns $h_{\mathcal{P}(q)}(p)$, or with traceback, $\sigma_{\mathcal{P}(q)}(p)$, and we wish to determine if $c(q, s) \in \partial P$. Thus, in a more general setting, we ask the following question called Relative position problem.

\begin{problem} (Relative position problem)
	\label{relative-position-problem}
	Given a polytope $P \subseteq \mathbb{R}^{d+1}$ represented by its supporting function $h_P$, its extremal function $\sigma_P$, and a point $x \in \mathbb{R}^{d+1}$. Determine if $x \not\in P$, $x \in \mathring{P}$, or $x \in \partial P$.
\end{problem}
Note that by the relationship between Optimisation oracle and Separation oracle, we can effectively determine if $x \in P$ in weakly polynomial time (cf. Appendix \ref{appendix: linear programming}), and thus hereinafter we shall assume $x \in P$.

Back to RNA setting, suppose $c(q, s)$ is learnable for $E$ with parameter $p$, we can also ask how \emph{robust} $p$ is, meaning how drastic one must change $p$ so that $s$ is no longer an optimal structure. To briefly see why this is relevant, it is believed that the environment surrounding an RNA can alter how the nucleotides pair with each other, effectively modifying the parameter of energy model. In that sense, robustness measures how stable an RNA structure is in changes of the surrounding: a stable RNA is more resistant to mutations and undesirable functions or defects thereof, which is a desirable property and supported by evolution.

To state formally, we seek to study the following question.

\begin{problem} (Robustness problem)
	\label{robustness-problem}
	Given a sequence $q$ and a structure $s$ learnable for an energy model $E$ with a parameter set $p$. Determine the robustness of $p$.
\end{problem}

\subsection{Related works on Relative position problem}

It turns out that Problem \ref{relative-position-problem} is fundamental in computation geometry, by its relation with collision detection problem, stated as follow:
\begin{quote}
	Given two polytopes $A$ and $B$ in $\mathbb{R}^{d+1}$, determine if they are disjoint.
\end{quote}

Gilbert, Johnson, and Keerthi \cite{Gilbert1988} showed that this problem can be reduced to determining of $A \ominus B$ contains the origin, where $A \ominus B = \{a - b \mid (a, b) \in A \times B\}$ denotes the Minkowski difference. They also showed that a given $p \in \mathbb{R}^{d+1}$, one has $h_{A \ominus B} (p) = h_A (p) - h_B(-p)$, thus although the Minkowski difference, much like the Minkowski sum, can be complicated (in the sense of high complexity), its supporting function is easy to compute. Finally, they demonstrated an algorithm for the following problem
\begin{quote}
	Given a polytope $P$ represented by its supporting function, determine if $P$ contains the origin.
\end{quote}
The key idea is Carathéodory's theorem, which states that for a given polytope $P \subseteq \mathbb{R}^{d+1}$, any point $x \in P$ can be written as a linear combination of $d+2$ points of $P$ (which one can relax to $d+2$ vertices of $P$). Their method, so-called GJK algorithm, aimed to find such a simplex, whose vertices lie on the boundary of $P$, that contains $x$. If no such simplexes can be found, $x \not\in P$, otherwise it is easy to check if x lies in the interior or on the boundary of $P$.

For our interest when the dimension is high, their method, or rather the subroutine, called Johnson's distance subalgorithm, to compute the distance from a point $x$ to a $d+1$-simplex on which relies GJK algorithm, has two main weaknesses.
\begin{itemize}
	\item Firstly, it is not numerical stable enough, which was the reason why in the original paper \cite{Gilbert1988}, the authors introduced a backup procedure.
	\item Secondly, and more importantly, for polytopes in $\mathbb{R}^{d+1}$, it requires computing the distance from $x$ to affine subspaces generated by all $2^{d+2}-1$ combinations of vertices.
\end{itemize}
There have been attempts to mitigate one or both the aforementioned issues: in particular, Cameron \cite{Cameron1997} modified the order in which the combinations are checked, thus improving the performance in practice, but without any theoretical bound. More recently, Montanari et al. \cite{Montanari2017} decreased the number of combinations to be checked to $2^{d+1}$, but no methods which require only a polynomial (with respect to $d$) number of checks are known. For most problems where GJK algorithm is used, $d$ is often fixed and small, typical $d = 2$, and thus one can afford such a number of checks. Unfortunately, our problem has high dimension, thus GJK algorithm is not applicable.

Snethen \cite{Snethen2008} proposed another algorithm, called Minkowski Portal Refinement, hereinafter abbreviated as MPR algorithm, to the same problem, which relies on the same principle, but with some modifications which allow reducing the number of checks from $2^{d+1}$ to $d+1$.

To summarise the idea, one first finds and fixes a point $o$ in the interior, then find a \emph{portal} defined as a $d$-simplex $P'$ through $d$ points on the boundary, which oftentimes are some vertices of $P$. Such a portal is said to be satisfying if the line segment $[o, x]_{\mathbb{R}^{d+1}}$ intersects $P'$. If that be the case, let $x'$ be the orthogonal projection of $x$ onto $P'$, the point $y = \sigma_P(x - x')$ forms with $P'$ a $d+1$-simplex. Since $[o, x]$ passes through $P'$, if $x$ does not lie in the interior of $P'$, $[o, x]$ must pass through another facet of this $d+1$-simplex that contains $y$. This new facet plays the role of the new portal, and the algorithm continues to termination.

In the original paper \cite{Snethen2008}, Snethen did not specify a definite method to find an initial portal. He proposed fixing the point $o$ first, then finding a portal, but this does not guarantee to work on the first attempt, and may need multiple tries. In Chapter \ref{chapter: simplex method part 2}, we propose another approach, where we find a non-degenerate $d+1$-simplex $Q$ whose vertices lie on the boundary of $P$ first, then choose $o$ as a point lying in the interior of $Q$, e.g. the centroid of $Q$ as it is convex.

Strangely, MPR algorithm is little mentioned in the literature, to which Neumayr and Otter \cite{Neumayr2017} addressed by describing an improved version with handling of termination conditions and special cases in 3-dimensional space. To the extend of our knowledge, neither full descriptions in general dimension nor proofs of termination for MPR are known. Chapter \ref{chapter: simplex method part 2} shall present more detailed description of MPR algorithm in the context of our scheme, and aim at addressing the aforementioned issues.

Recently, Hornus \cite{Hornus2017} demonstrated another approach to the problem with a method called Decision Sphere Search, hereinafter abbreviated as DSS algorithm.

To summarise, suppose $x$ lies on the boundary of $P$, and let $S$ be its corresponding normal spherical polytope. In one iteration, one feeds a vector $p \in \mathbb{S}^d$ and obtains the great circle defined by $C = \mathbb{S}^d \cap \{y \mid \langle y, p \rangle = \langle x, p \rangle \}$ which divides $\mathbb{S}^d$ into two halfspheres, namely $C_{+} = \mathbb{S}^d \cap \{y \mid \langle y, p \rangle > \langle x, p \rangle \}$ and $C_{-} = \mathbb{S}^d \cap \{y \mid \langle y, p \rangle < \langle x, p \rangle \}$. If $h_P(p) = \langle x, p \rangle$, we conclude that $x \in \partial P$. Else, $S$ must lie in the halfsphere $C_{-}$, whose centre is given by $p' = \frac{x - \sigma_P(p)}{\| x - \sigma_P(p) \|}$. This new vector $p'$ is then fed into the next iteration.

The collection of halfspheres from the execution defines a spherical polytope $S'$ that bounds $S$, i.e $S \subseteq S'$. At some point, either we come across a vector $p \in S$ which proves $x \in \partial P$, or prove that $S$ is necessarily empty by showing that $S'$ is empty. This approach has its own limitations, amongst which is the lack of bound on complexity that is independent of the volume $\vol{S}$ of $S$, which we shall overcome in Chapter \ref{chapter: ellipsoid method part 2} with ideas from ellipsoid method in linear programming.
\section{Robustness problem: telescoping method}
\label{chapter: telescoping method part 2}

\subsection{Introduction}

In this chapter, we present a method whereby, similarly to modification of the original dynamic programming algorithm from tropical algebra to polytope algebra as presented by Pachter and Sturmfels \cite{Pachter2004, Pachter2005}, one can modify the algorithm further to retrieve only a certain part of the polytope $\mathcal{P}(q)$. The key idea is that for any two polytopes $A$ and $B$, a given quadrant of $A \oplus B$ or of $A \otimes B$ is determined only by the corresponding quadrants of $A$ and $B$, as demonstrated in Figure \ref{fig: restricted polytope algebra}. This is stated formally as follow.

\begin{figure}[h]
	\centering
	\subfloat[\centering $F_D(A)$]{{\includegraphics[width=0.2\textwidth]{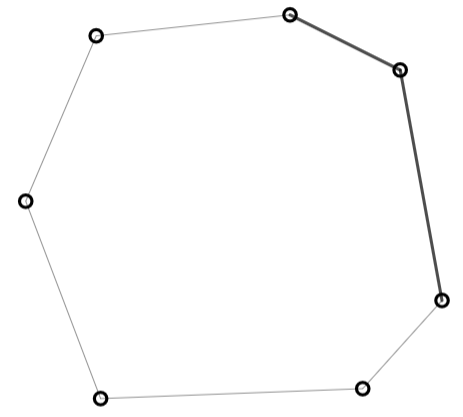} }}%
	\qquad
	\subfloat[\centering $F_D(B)$]{{\includegraphics[width=0.2\textwidth]{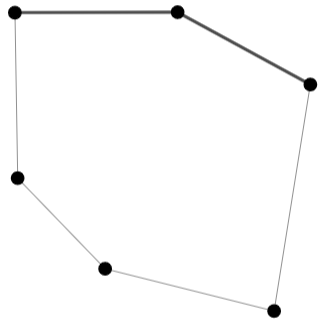} }}%
	\qquad
	\subfloat[\centering $F_D(A \oplus B)$]{{\includegraphics[width=0.2\textwidth]{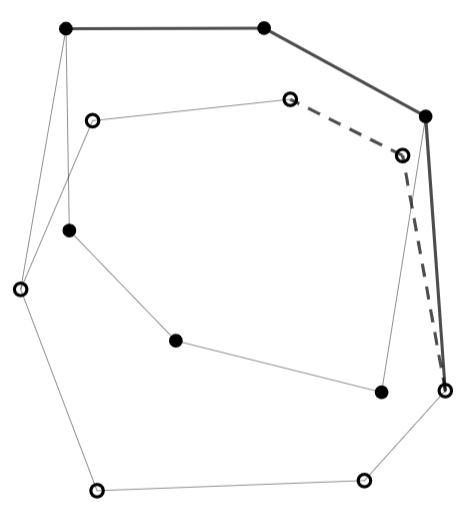} }}%
	\qquad
	\subfloat[\centering $F_D(A \otimes B)$]{{\includegraphics[width=0.2\textwidth]{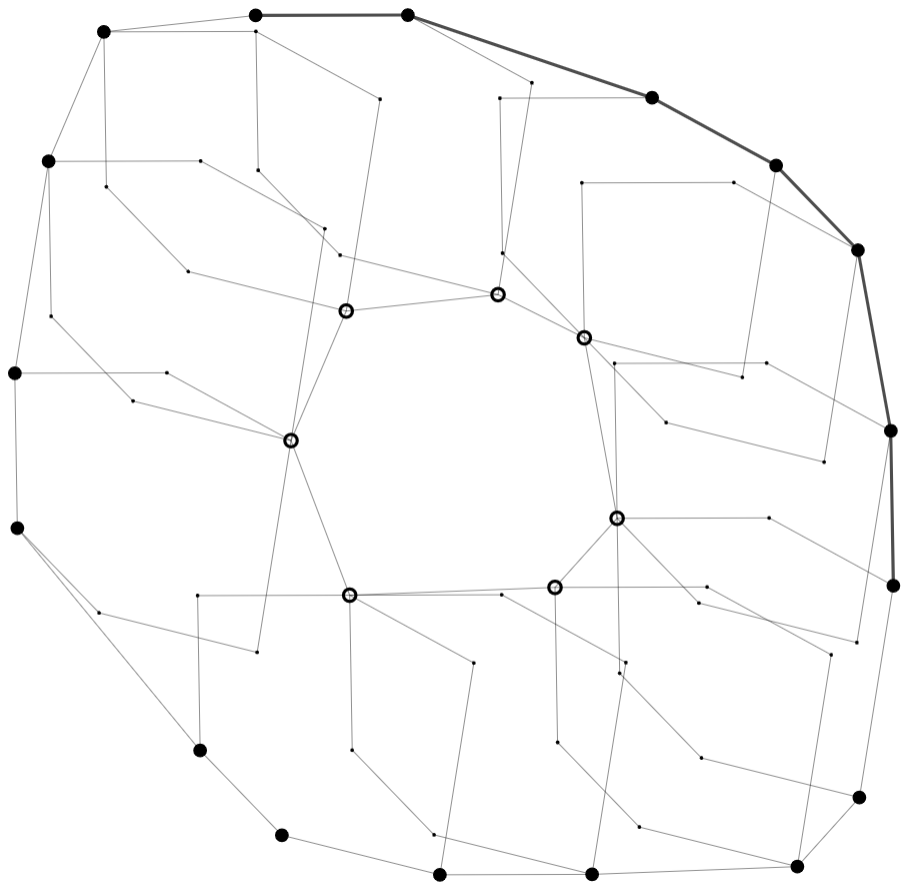} }}%
	\caption{Example of operations in polytope algebra restricted to D being the first quadrant.}
	\label{fig: restricted polytope algebra}
\end{figure}

\begin{theorem}
	\label{theorem on quadrant}
	Let $D \subseteq \mathbb{R}^{d+1}$, $A$ and $B$ be two polytopes in $\mathbb{R}^{d+1}$. For a given polytope $P \in \mathbb{R}^{d+1}$, we define 
	\[
	F_D(P) = \{x \mid \exists d \in D, \langle x, d\rangle = h_P(d)\}.
	\]
	Then, one has $F_D(A \oplus B) \subseteq F_D(A) \oplus F_D(B)$, and $F_D(A \otimes B) \subseteq F_D(A) \otimes F_D(B)$.
\end{theorem}
\begin{proof}
	Let $d \in D$, one has 
	\[
	\max_{x \in A \otimes B} \langle x, d \rangle = \max_{(a, b) \in A \times B} \langle a + b, d \rangle = \max_{(a, b) \in A \times B} \langle a, d \rangle + \langle b, d \rangle = \max_{a \in A} \langle a, d \rangle + \max_{b \in B} \langle b, d \rangle\\
	\]
	which implies that for $x = a + b \in A \otimes B$ for some $a \in A$ and $b \in B$, then if $x \in F_D(A \otimes B)$ then $a \in F_D(A)$ and $b \in F_D(B)$, implying $x \in F_D(A) \otimes F_D(B)$, as desired.
	
	On the other hand, by definition, for any $x \in A \oplus B$, there exist $\lambda, \mu \geq 0$ and $a \in A$, $b \in B$, such that $\lambda + \mu = 1$ and $x = \lambda a + \mu b$. Now suppose $\lambda > 0$, and $\langle x, d \rangle = h_{A \oplus B}(d)$, if there exists $a' \in A$ such that $\langle a', d \rangle > \langle a, d \rangle$, then for $x' = \lambda a' + \mu b \in A \oplus B$, one has
	\[
	\langle x, d \rangle = \lambda \langle a, d \rangle + \mu \langle b, d \rangle < \lambda \langle a, d \rangle + \mu \langle b, d \rangle = \langle x', d \rangle, 
	\]
	a contradiction, thus $a \in F_D(A)$. Similarly, if $\mu > 0$ then $b \in F_D(B)$, which together implies $x \in F_D(A) \oplus F_D(B)$.
\end{proof}

This results allows one to specify beforehand a quadrant $D$ of interest, and as they carry out the modified dynamic programming algorithm according to Pachter and Sturmfels \cite{Pachter2004, Pachter2005}, only retain relevant parts of the polytope, thus significantly reduce the complexity of the polytopes and speed up the computation.

\subsection{Toward robustness}

Now suppose we know a pair $(q, s)$ is learnable for a parameter vector $p$. To determine the robustness of $p$, we shall consider $D$ to be of some special forms: in particular, we restrict $D$ to $\mathbb{S}^d$, and consider $D$ to be a closed ball $B_{\mathbb{S}^d}(p, \theta) = \{x \in \mathbb{S}^d \mid d_{\mathbb{S}^d}(x, p) \leq \theta\}$ for some $\theta > 0$. Let $F$ be a polytope of $\mathcal{S}(\mathcal{P}(p))$ of minimum dimension that contains $p$, then the robustness of $p$ is the minimum $\theta$ such that $B_{\mathbb{S}^d}(p, \theta)$ intersects another face $F'$ of $\mathcal{S}(\mathcal{P}(p))$, which one can limit to only the adjacent faces of $F$. Or, equivalently, it is the infimum of $\theta$ such that $F$ contains $B_{\mathbb{S}^d}(p, \theta)$.

Note that we know a priori that $0 \leq \theta_P(p) \leq \pi$, thus we can compute $\theta_P(p)$ via binary search. But it is also possible to compute the robustness with only one execution of the dynamic programming algorithm, and no binary search is necessary. Similar to the above, we shall modify the dynamic programming algorithm once more, as follow.

One crucial observation is that $D$ needs not be fixed throughout the dynamic programming process, but in fact can be decreased, i.e. the new set $D'$ for the next iteration of the dynamic programming needs only be a subset of $D$, since by definition, if $D' \subseteq D$, then for any polytope $P$, $F_{D'}(P) \subseteq F_D(P)$. This gives us the idea to set initially $\theta = \pi$ and $D = B_{\mathbb{S}^d}(p, \theta) = \mathbb{S}^d$, and as we carry out the dynamic programming with $F_D$, we may decrease $\theta$ until $\theta$ equals $\theta_{\mathcal{P}(q)}(p)$. 

There remains one last question: can we decrease $\theta$? Or, in other words, given $F_{B_{\mathbb{S}^d}(p, \theta_A(p))}(A)$ and $F_{B_{\mathbb{S}^d}(p, \theta_B(p))}(B)$, can we construct $F_{B_{\mathbb{S}^d}(p, \theta_{A \oplus B}(p))}(A \oplus B)$ and $F_{B_{\mathbb{S}^d}(p, \theta_{A \otimes B}(p))}(A \otimes B)$? The answer is yes, and the following theorem provides rigorous reasoning.

\begin{theorem}
	\label{theorem on decreasing balls}
	Let $A, B \subseteq \mathbb{R}^{d+1}$ be two polytopes, and $p \in \mathbb{R}^{d+1}$. With notations defined as in Theorem \ref{theorem on quadrant}, one has 
	\[
		F_{B_{\mathbb{S}^d}(p, \theta_{A \oplus B}(p))}(A \oplus B) \subseteq F_{B_{\mathbb{S}^d}(p, \theta_A(p))}(A) \oplus F_{B_{\mathbb{S}^d}(p, \theta_B(p))}(B)
	\]
	and
	\[
		F_{B_{\mathbb{S}^d}(p, \theta_{A \otimes B}(p))}(A \otimes B) \subseteq F_{B_{\mathbb{S}^d}(p, \theta_A(p))}(A) \otimes F_{B_{\mathbb{S}^d}(p, \theta_B(p))}(B).
	\]
\end{theorem}
\begin{proof}
	We recall Proposition 7.12 in Ziegler's \emph{Lectures on Polytopes} \cite{Ziegler1995} stating that $\mathcal{N}(A \otimes B) = \mathcal{N}(A) \wedge \mathcal{N}(B)$, where $\cdot \wedge \cdot$ denotes the common refinement. Thus, it follows immediately by definition that $\theta_{A \otimes B} (p) = \min (\theta_A(p), \theta_B(p))$, whence together with Theorem \ref{theorem on quadrant}, one deduces
	\[
	\begin{split}
		F_{B_{\mathbb{S}^d}(p, \theta_{A \otimes B}(p))}(A \otimes B)
		& \subseteq F_{B_{\mathbb{S}^d}(p, \theta_{A \otimes B}(p))}(A) \otimes F_{B_{\mathbb{S}^d}(p, \theta_{A \otimes B}(p))}(B) \\
		& \subseteq F_{B_{\mathbb{S}^d}(p, \theta_A(p))}(A) \otimes F_{B_{\mathbb{S}^d}(p, \theta_B(p))}(B).
	\end{split}
	\]
	
	Now let $x \in F_{B_{\mathbb{S}^d}(p, \theta_{A \oplus B}(p))}(A \oplus B)$, and let $q \in B_{\mathbb{S}^d}(p, \theta_{A \oplus B}(p))$ such that $\langle x, q \rangle = h_{A \oplus B}(q)$. We notice that for any $r \in \mathbb{R}^{d+1}$, one has $h_{A \oplus B}(r) = \max(h_A(r), h_B(r))$. Suppose without loss of generality that $x \in A$, then it follows that $h_A(q) \geq h_B(q)$, otherwise one would have $\max(h_A(q), h_B(q)) = \langle x, q \rangle \leq h_A(q) < h_B(q)$,
	a contradiction.
	
	Suppose by contradiction that $h_A(p) < h_B(p)$. Let $y = \sigma_A(q)$ then $\langle y, q \rangle = h_{A \oplus B}(q)$. Since $q \in B_{\mathbb{S}^d}(p, \theta_{A \oplus B}(p))$, one has $h_A(p) \geq \langle y, p \rangle = h_{A \oplus B}(p) = h_B(p) > h_A(p)$, a contradiction. Hence, $h_A(p) \geq h_B(p)$.
	
	Then for any $z \in A$, if $\langle z, p \rangle = h_A(p)$, one has $\langle z, p \rangle = h_{A \oplus B}(p)$, and by assumption $q \in B_{\mathbb{S}^d}(p, \theta_{A \oplus B}(p))$, it follows that $\langle z, q \rangle = h_{A \oplus B}(q) = h_A(q)$. Or, equivalently, $q \in B_{\mathbb{S}^d}(p, \theta_A(p))$, which implies $x \in F_{B_{\mathbb{S}^d}(p, \theta_A(p))}(A)$. Then, it follows that
	\[
	F_{B_{\mathbb{S}^d}(p, \theta_{A \oplus B}(p))}(A \oplus B) \subseteq F_{B_{\mathbb{S}^d}(p, \theta_A(p))}(A) \oplus F_{B_{\mathbb{S}^d}(p, \theta_B(p))}(B).
	\]
\end{proof}

\subsection{Computing robustness given a polytope}

For a given polytope $P$ in the course of executing the dynamic programming, Theorem \ref{theorem on decreasing balls} allows one to consider only $F_{B_{\mathbb{S}^d}(p, \theta_P(p))}(P)$, but does not explicitly say how to compute it. Whilst it is possible via constructing the normal fan, we remind that translating from $\mathcal{V}$- to $\mathcal{H}$-representation introduces a bottleneck and is responsible for the complexity of the naive approach. For this reason, we introduce the following theorem which characterises $D = B_{\mathbb{S}^d}(p, \theta)$ when $\theta = \theta_{\mathcal{P}(q)}(p)$. 

\begin{theorem}
	\label{theorem on distance to the boundary}
	Let $P \subseteq \mathbb{S}^d$ be a polytope, $p \in \mathring{P}$, and $p^* \in \partial P$ such that $d_{\mathbb{S}^d}(p, p^*) = d_{\mathbb{S}^d}(p, \partial P) = \min_{p' \in \partial P} d_{\mathbb{S}^d}(p, p')$. Then $p^*$ lies in the interior of a facet of $P$, i.e. no faces of $P$ with lower dimension contain $p^*$.
\end{theorem}

In simple terms, if the robustness of $p$ is positive, one knows that $p$ lies in the interior of some normal spherical polytope $S \in \mathcal{S}(P)$, which corresponds to one unique vertex $x \in P$ such that $\langle x, p \rangle = h_P(p)$, i.e. $x = \sigma_P(p)$.. Theorem \ref{theorem on distance to the boundary} tells us that the robustness $\theta_P(p)$ is obtained by an edge between $x$ and a neighbour vertex $y$ of $P$, or more precisely, if we denote $H = \{z \in P | \langle z, p \rangle = 0\}$ to be the hyperplane admitting $p$ as its normal vector, then
\[
\theta_P(p) = \min_{y \in \vertex{P}\setminus\{x\}} \angle(y - x, H)= \min_{y \in \vertex{P}\setminus\{x\}} \sin^{-1} \frac{(h_P(p) - \langle y, p \rangle)\| p \|}{\| x - y\|}.
\]
Thus, one only needs to find \emph{all} points $y \in \vertex{P}\setminus\{x\}$ that minimise $\angle(y - x, H)$, and compute $\theta_P(p)$ without explicitly constructing $\mathcal{N}(P)$. This allows us to restrict ourselves to $\mathcal{V}$-representation.
\section{Relative position problem: a practical solution}
\label{chapter: simplex method part 2}

In this chapter, we describe a mathematical framework which incorporates MPR algorithm in general dimension, and opens the exploration of other alternatives. In particular, we explore some of such variants, and demonstrate that they can perform comparably to MPR algorithm. As we shall see, in the context of relative position problem, we shall find a notion similar to pivot rules for the simplex method and for which we shall coin the same name. Further exploration and experiments on some pivot rules show a phenomenon similar to simplex method in linear programming: performance depends on the choice of rule, and we shall construct an example where all rules lead to linear complexity with respect to the number of vertices of the polytope.

In this Section, we consider the case where the polytope $P$ is non-degenerate. Otherwise, $P$ has empty interior, so $P = \partial P$, and if $x \in P$, we return $x \in \partial P$.

\subsection{Method description}

\subsubsection{Geometric idea}

The key idea of our method is as follow: let $P \in \mathbb{R}^{d+1}$ be a polytope given by its supporting function $h_P$ and extremal function $\sigma_P$, and a point $x \in P$. Suppose we have $d+2$ vertices of $P$, denoted by $x_0, x_1, ..., x_{d+1}$, that forms a non-degenerate simplex $Q$, then
\begin{itemize}
	\item either $x$ lies in the interior of $Q$, at which point we conclude as $x \in \mathring{Q} \subseteq \mathring{P}$; 
	\item or $x$ lies on the boundary of $Q$, at which point we conclude that $x \in \partial P$ if $x$ coincides with one of $d+1$ vertices, and $x \in \mathring{P}$ otherwise; 
	\item or $x$ lies outside $Q$, at which point there exist $d+1$ vertices whereby the defined (affine) hyperplane $H$ separates $x$ and the other vertex. 
\end{itemize}
If the first two cases do not apply, then we need to check only all $d+1$ subsets of $d+1$ from $d+2$ vertices. Suppose $x$ lies on the boundary of $Q$ or outside $Q$, and by possibly a re-indexing, suppose that $x_1, x_2, ..., x_{d+1}$ defines an affine hyperplane $H$ that separate $x$ and $x_0$. Let $h$ be a normal vector of $H$, and suppose $\langle x, h\rangle > 0$, which implies $\langle x_0, h \rangle < 0$, then we replace $x_0$ with $\sigma_P(h)$.

\begin{figure}[h]
\centering
\begin{minipage}{.45\textwidth}
	\centering
	\includegraphics[width=\linewidth]{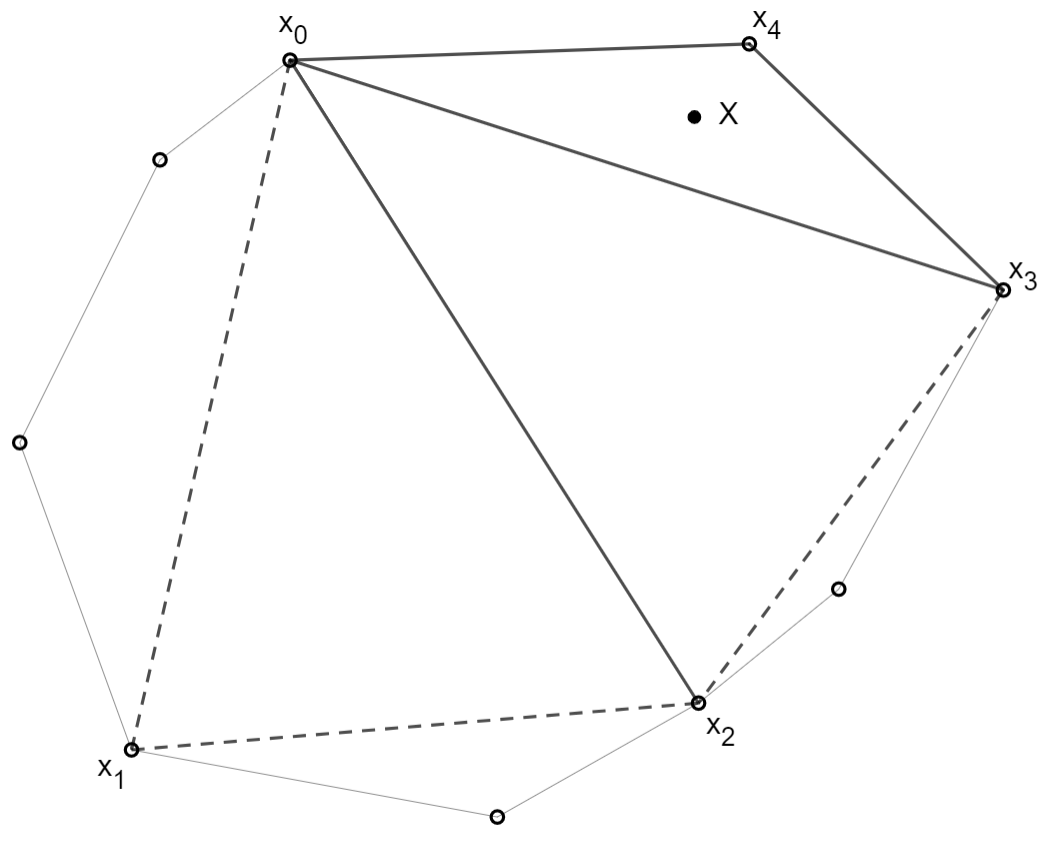}
	\captionof{figure}{Example in 2D.}
	\label{fig:test1}
\end{minipage}%
\begin{minipage}{.45\textwidth}
	\centering
	\includegraphics[width=\linewidth]{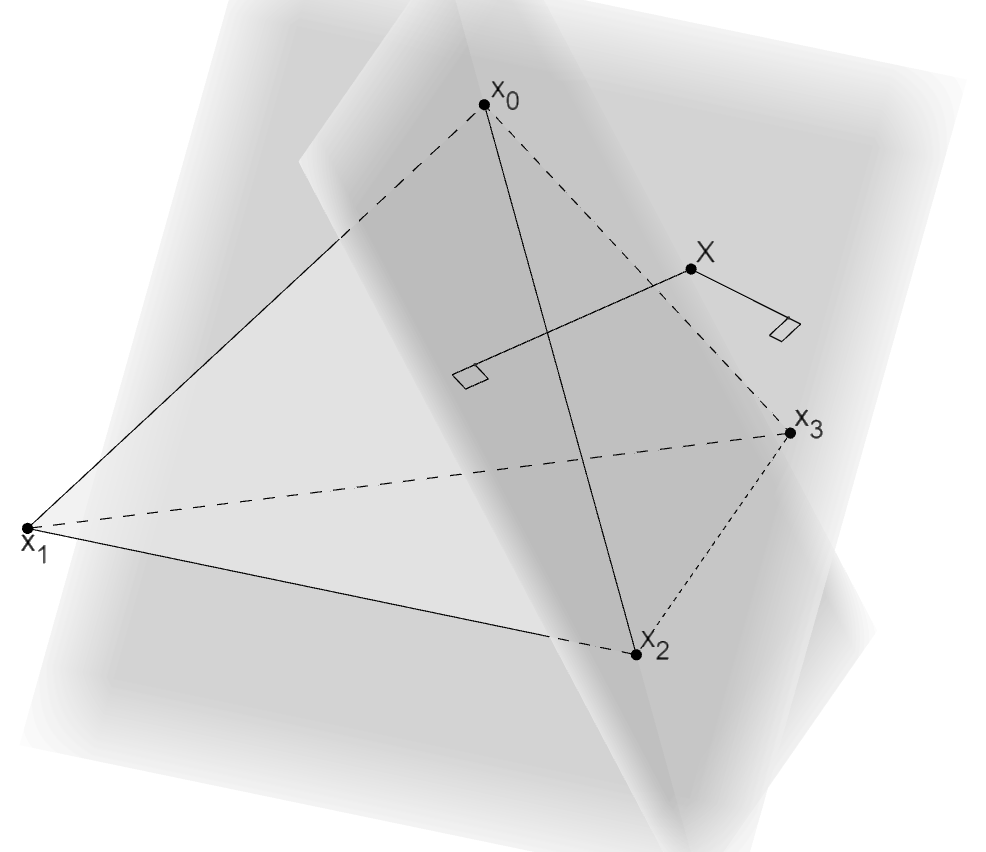}
	\captionsetup{justification=centering}
	\captionof{figure}{A case where multiple hyperplanes are possible.}
	\label{fig:test2}
\end{minipage}
\end{figure}

In fact, one can perform all $d+1$ checks simultaneously as follow: since $Q$ is non-degenerate, $d+1$ vectors of the form $v_i - v_0$ for all $i = 1, 2, ..., d+1$ forms a basis of $\mathbb{R}^{d+1}$, and one can find a unique set of coefficients $\lambda_1, \lambda_2, ..., \lambda_{d+1}$ such that $x = \sum_{i = 1}^{d+1}
\lambda_i (x_i - x_0)$. Then
\begin{itemize}
	\item if $\lambda_i < 0$ for some $1 \leq i \leq d+1$, then $x_1, x_2, ..., x_{i-1}, x{i+1}, ..., x_{d+1}$ define an affine hyperplane $H_i$ separating $x$ and $x_i$.
	\item if $\sum_{i = 1}^{d+1} \lambda_i > 1$, then $x_2, x_3, ..., x_{d+1}$ define an affine hyperplane $H_0$ separating $x$ and $x_0$.
	\item otherwise, one will have $\lambda i \geq 0$ for all $i$ and $\sum_{i = 1}^{d+1} \lambda_i \leq 1$, which proves that $x \in Q$.
\end{itemize}

The procedure is repeated until termination, at which point we may conclude definitively: if $x \not\in Q$ and no further progress can be made, then it is the case that $H$ separates $x$ and $P$, thus $x \not\in P$.

We remark the difference between this scheme and GJK algorithm: in particular, we shall only replace \emph{one} vertex, and thus ensure that once we construct a non-degenerate simplex, the simplices in subsequent iteration will always be non-degenerate. In GJK algorithm, we keep only the minimum set of vertices whose affine hull contain the points in the simplex that is closest to $x$, which may result in removal of one or more points.

This raises two questions: \emph{Does this process always terminate?} and \emph{How long will it run?} It turns out that both questions rely on the choice of pivot rule, i.e. a rule to decide which vertex to be replace if there exist multiple such separating hyperplanes $H$. For now, we shall split the algorithm into two phases: Phase 1 to find an initial simplex, and Phase 2 to refine the simplex, which concerns the pivot rules.

\subsubsection{Phase 1: Finding an initial simplex}
This phase consists of $d+2$ iterations, and can be described as follow.
\begin{enumerate}
	\item Consider a random direction $u_0$, and let $v_0 = \sigma_P(u_0)$.
	\item At $i$\textsuperscript{th}, consider the affine hull $\mathcal{H}$ of points $v_0, ..., v_{i-1}$. We chose an arbitrary normal vector $u_i$ of $\mathcal{H}$ and compute $h_P(u_i)$ Note that by construction, we have $\langle y, u_i \rangle \leq h_P(u_i)$ for all $y \in \mathcal{H}$. We have two cases:
	\begin{itemize}
		\item If $\langle v_0, u_i \rangle = h_P(u_i)$, then since $P$ is non-degenerate, one must have that $\langle v_0, -u_i \rangle < h_P(-u_i)$. Thus, let $v_i = \sigma_P(-u_i)$, for all $j = 0, 1, .., i-1$, one has $\langle v_j, -u_i \rangle < \langle v_i, -u_i \rangle$ . 
		\item If $\langle v_0, u_i \rangle < h_P(u_i)$, then let $v_i = \sigma_P(u_i)$, one has $\langle v_j, u_i \rangle < \langle v_i, u_i \rangle$ for all $j = 0, 1, .., i-1$.
	\end{itemize}
\end{enumerate}
In any cases, the affine hull $\mathcal{H}$ in the next iteration will increase in dimension, thus after $d+2$ iterations, we obtain a non-degenerate simplex $Q$.

In fact, we can incorporate the information regarding $x$ into the scheme, which leads to potential early termination, demonstrated as follow:
\begin{enumerate}
	\item Consider a random direction $u_0$, and let $v_0 = \sigma_P(u_0)$.
	\item At $i$\textsuperscript{th}, consider the affine hull $\mathcal{H}$ of points $v_0, ..., v_{i-1}$. We have two case:
	\begin{enumerate}
		\item 
		$x \not\in \mathcal{H}$: Let $x_\mathcal{H}$ be the projection of $x$ onto $\mathcal{H}$, $u_i = x - x_\mathcal{H}$, and compute $h_P(u_i)$. Note that by construction, we have $\langle y, u_i \rangle < \langle x, u_i \rangle$ for all $y \in \mathcal{H}$. We have two subcases.
		\begin{itemize}
			\item If $\langle x, u_i \rangle > h_P(c)$, then $x \not\in P$ for the affine hyperplane passing through $v_i$ and admitting $c$ as a normal vector separates $P$ and $i$.
			\item If $\langle x, u_i \rangle \leq h_P(c)$, then let $v_i = \sigma_P(u_i)$, one has $\langle v_j, u_i \rangle < \langle v_i, u_i \rangle$ for all $j = 0, 1, .., i-1$. The affine hull $\mathcal{H}$ in the next iteration will increase in dimension. 
		\end{itemize}
		\item
		$x \in \mathcal{H}$: We chose an arbitrary normal vector $u_i$ of $\mathcal{H}$ and compute $h_P(u_i)$ Note that by construction, we have $\langle y, u_i \rangle \leq h_P(u_i)$ for all $y \in \mathcal{H}$. We have two subcases:
		\begin{itemize}
			\item If $\langle x, u_i \rangle = h_P(u_i)$, then if $x \in \conv{\{v_0, v_1, ..., v_{i-1}\}}$, one has that $x \in \partial P$. Otherwise, we proceed as above by letting $v_i = \sigma_P(-u_i)$.
			\item If $\langle x, u_i \rangle < h_P(u_i)$, then we proceed as above by let $v_i = \sigma_P(u_i)$.
		\end{itemize}
		The affine hull $\mathcal{H}$ in the next iteration will increase in dimension.
	\end{enumerate}
\end{enumerate}
Since after each iteration, either the algorithm terminates or the dimension of $\mathcal{H}$ increases, thus by the end, we will obtain a non-degenerate simplex $Q$. This variant of Phase 1 requires more checking, which involves more linear algebra operations, and thus we shall not consider in our implementation.

\subsubsection{Phase 2: Refining the simplex}

Now we have our initial simplex $Q$, let $I$ be the set of indices $i$ such that $x_0, x_1, ..., x_{i-1}, x_{i+1},..., x_{d+1}$ defines an affine hyperplane $H_i$ that separate $x$ and $x_i$. We assume that $I$ is non-empty, otherwise the algorithm terminates.

If we have $|I| = 1$, i.e. there exists uniquely such a vertex., then the construction of simplex for the next iteration is unambiguous. Unfortunately, it is sometimes, if not often, the case that $|I| > 1$, i.e. there exist multiple such vertices, and we have to pick one of them according to some pivot rules. Here, we outline three possible variants, amongst many:
\begin{enumerate}
	\item (Randomized Rule, abbreviated as RR) Pick a random index from $I$.
	\item (Farthest Hyperplane Rule, abbreviated as FHR) Pick index $i$ that \emph{maximises} $\|x - x_i\|$ where $x_i$ is the projection from $x$ to $H_i$.
	\item (Closest Hyperplane Rule, abbreviated as CHR) Pick index $i$ that \emph{minimises} $\|x - x_i\|$ where $x_i$ is the projection from $x$ to $H_i$.
\end{enumerate}

Notably, in this framework, there is only particular pivot rule (or, to be more precise, a family whereof) that corresponds to MPR algorithm, describe as follow:
\begin{enumerate}
	\item Initially, after Phase 1, we choose a point $o$ that we know to be in the interior of $Q$ and thus of $P$ (a good candidate is the centroid of $Q$). If the algorithm has yet to terminate, $x$ will lie outside $Q$, hence the segment $[o, x]$ intersects one facet of $Q$. It is thus necessary that the affine hyperplane containing this facet separate $x$ from $o$, and therefore from the remaining vertex of $Q$. This facet is called a portal, and in this case we will denote by $R_0$. If many such facets exist, we choose one arbitrarily.
	\item Then, at $j$\textsuperscript{th} iteration in Phase 2, given $v_0, v_1, ..., v_{d+1}$ whose convex hull $R_{j-1} = \conv{\{v_0, v_1, ..., v_{d+1}\}}$ defines the portal, let $x'$ be the projection of $x$ onto the affine hyperplane $H$ containing the portal, $u = x - x'$, and $v_{d+2} = \sigma(u)$. If the algorithm does not terminate, then $x$ lies outside $Q = \conv{\{v_0, v_1, ..., v_{d+2}\}}$, and since $[o, x]$ intersects $R_{j-1}$, there must exist at least another facet of $Q$ given by $x_{d+2}$ and some $d$ points amongst $v_0, v_1, ..., v_{d+1}$. If many such facets exist, we choose one arbitrarily, this defines our new portal $R_j$.
\end{enumerate}

\begin{theorem}
	MPR algorithm always terminates.
\end{theorem}
\begin{proof}
	It suffices to show an invariant that is (A) bounded, and (B) strictly increasing or decreasing after each iteration. Many choices are possible, but here we choose the following: let $x_j = R_j \cap [o, x]$, whose existence is guaranteed by the method. For (A), one has $x_j \in [o, x]$, so in particular, $0 \leq \|x_j - o \| \leq \|x - o\|$. As for (B), at $j$\textsuperscript{th} iteration, if the algorithm does not terminate, then $v_{d+2}$ lies strictly in the open halfspace defined by $H$ not containing $o$, and thus so is $x_j$. Therefore, one has that $x_{j-1} \in [o, x_j]$ and $\|x_j - o \| > \| x_{j-1} - o \|$, as desired.
\end{proof}

\subsection{Experiments}

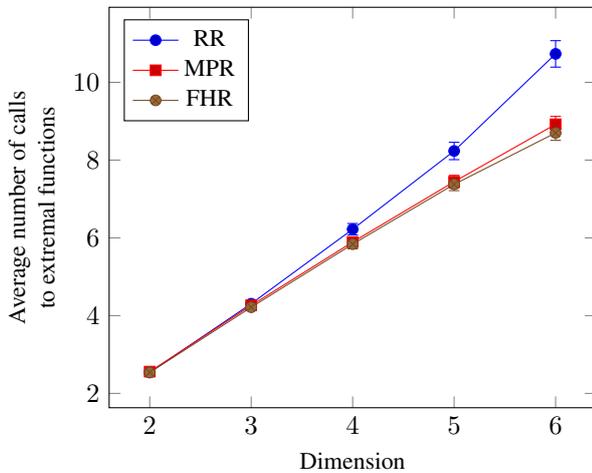
\begin{wrapfigure}[17]{l}{0.47\textwidth}
	\vspace{-12pt}
	\centering
	\begin{tikzpicture}
		\begin{axis}[
			width=0.47\textwidth, 
			xlabel = Dimension, 
			ylabel = Average number of calls\\ to extremal functions,
			ylabel style={align=center}, 
			label style={font=\small}, 
			legend style={font=\small},
			legend pos=north west, 
			xtick distance=1
			]
			\addplot plot [error bars/.cd, y dir=both, y explicit]	table [x index = 0, y index = 1, y error minus index = 2, y error plus index = 3] {simplex_experiment_data.txt};
			\addlegendentry{RR}
			\addplot plot [error bars/.cd, y dir=both, y explicit]	table [x index = 0, y index = 4, y error minus index = 5, y error plus index = 6] {simplex_experiment_data.txt};
			\addlegendentry{MPR}
			\addplot plot [error bars/.cd, y dir=both, y explicit]	table [x index = 0, y index = 7, y error minus index = 8, y error plus index = 9] {simplex_experiment_data.txt};
			\addlegendentry{FHR}
		\end{axis}%
	\end{tikzpicture}
	\vspace{-10pt}
	\captionsetup{justification=centering}
	\captionof{figure}{Performance of the three pivot rules. The error bars represent 95\% confidence interval.}
	\label{figure: pivot rule performance}
\end{wrapfigure}

Here we consider the three pivot rules introduced above and the MPR algorithm. Except MPR algorithm, the implementation was done to allow one to give the pivot rule as a parameter, thus ensuring the same condition. Due to the point $o$, changes to the implementation were necessary for MPR algorithm. Given the similar linear algebraic operations involved, runtime proved to be too uncertain and imprecise as a measure of performance, and we used a different metric. As we can infer the supporting function from the extremal function, our measure is the number of calls made to the latter. In practice and especially for long sequences, it is often the case that the supporting function (and by extension, the extremal function) is the bottle neck.

For the test cases, we consider dimensions $2 \leq d \leq 6$. For each dimension, we generated 1000 polytopes, each of which was constructed as the convex hull of 1000 randomly generated points, uniformly distributed in the unit cube. For a given vector $p$, the supporting function $h_P$ was a loop over the set of vertices $y$, and the extremal function returns a point as an arbitrary linear combination of the vertices $y$ satisfying $\langle y, p \rangle = h_P(y)$. For each test case, we picked a random point on the boundary and a point in the interior to make sure the algorithms worked as expected.

During the experiment, we found that Closest Hyperplane Rule resulted in infinite cycles for some rare cases, and we therefore decided not to include its performance into the final comparison. For the other three rules, all the runtimes were of order milliseconds to seconds, negligible for practical purposes. Since we did not consider code optimisation, our observed runtimes will not reflect well the actual performance, which further inspired the need for another metric described as above. The result is displayed in Figure \ref{figure: pivot rule performance}. Note that Furthest Hyperplane Rule slightly outperformed MPR, whereas Random Rule performed worst than both, but they all called extremal function only a linear number of times.

To see how much FHR outperformed compared with MPR, the following table shows the percentage of tests where FHR performed better, as good, and worse than MPR in terms of number of calls to extremal functions. We can see that FHR often outperformed MPR, but with not significant percentage of the time. Moreover, the gain when FHR performs better is smaller than the loss when it does worse. Therefore, it is inconclusive which method is better. Nonetheless, we chose FHR for simplicity and not in looking for a better method. This result is promising and deserves closer studies in the future.

\begin{table}[ht]
	\centering
	\begin{tabular}[t]{lccccc}
		\toprule
		d       &  2     & 		3 & 	  4 &       5 & 	 6  \\
		\midrule
		Better  & 16.9\% & 18.6\% &  20.8\% & 21.65\% & 23.56\% \\
		As good & 67.2\% & 64.4\% & 59.65\% & 57.55\% & 56.85\% \\
		Worse   & 15.9\% & 17.0\% & 19.55\% &  20.9\% &  19.5\% \\
		\midrule
		$\Delta_{\text{better}}$ & 1.66 & 2.01 & 2.43 & 3.06 & 3.60 \\
		$\Delta_{\text{worse}}$  & 1.68 & 2.09 & 2.56 & 3.28 & 3.90\\
		\bottomrule
	\end{tabular}
	\captionsetup{justification=centering, margin=2cm}
	\caption{Comparison of performance for FHR and MPR. \\
		$\Delta_{\text{better}}$ and $\Delta_{\text{worse}}$ denote the gain and loss in number of extremal \\
		function calls when FHR performs better and worse, respectively. }
\end{table}%

\subsection{Complexity assuming termination}

From the mathematical description, it is clear that for $\sigma_P$ always returning a vertex of $P$ for any query, e.g. as an optimization oracle in linear programming problem, then suppose the algorithm terminates, the number of iterations is linear with respect to the number of vertices. Or in other words, if the algorithm terminates, one has that the complexity to be $O(|\vertex{P}|)$.

In this subsection, we will show that this bound is tight, and moreover, independent of the dimension. We shall assume the real computation model, i.e. numbers are represented with infinite precision, and for the sake of brevity, we shall outline only the key ideas of the construction.

The first key observation is how to construct such a polygon in 2-dimensional space. In this part, we shall use capital letters to denote points, and assume a plane without any coordinate systems.

Let $d$ be a line on which we choose an arbitrary point $A_0$. We take another point $A_1 \not\in d$, and draw the line $d_1$ passing through $A_1$ and perpendicular to $d$. Let $B_1$ be the intersection between $d$ and $d_1$. Then, on the halfplane defined by $d_1$ and containing $A_0$, we take another point $A_2$ such that the line segment $A_1 A_2$ intersects $A_0 B_1$. The triangle $A_0 A_1 A_2$ now gives our initial 2-simplex.

Then, let $d_2$ be the line parallel to $A_0 A_1$ and passing through $A_2$, and let $B_2$ be the intersection between $d_1$ and $d_2$. Since the line $d$ passes through the line segment $A_1 A_2$, it divides the triangle $\triangle A_1 A_2 B_2$ into two halves. We choose a point $A_3$ strictly in the interior of the half containing $A_1$.

And the process continues indefinitely. To be more precise, let $i \geq 2$, we draw the line $d_i$ passing through $A_i$ and parallel to $A_{i-2} A_{i-1}$, and let $B_i$ be the intersection between $d_{i-1}$ and $d_i$. Since the line $d$ passes through the line segment $A_{i-1} A_i$, it divides the triangle $\triangle A_{i-1} A_i B_i$ into two halves. We choose a point $A_{i+1}$ strictly in the interior of the half containing $A_{i-1}$.

Now let $n \geq 3$. we claim that 
\begin{enumerate}
	\item The points $A_0, A_1, ..., A_{n+1}$ are vertices of $P = \conv{\{A_0, A_1, ..., A_{n+1}\}}$.
	
	Indeed, for each $i > 0$, one has the line $d_i$ separates the plane into two halves, one of which contains $P$, and since $A_i \in d_i$, $A_i$ is a vertex of $P$. As for $A_0$, one can see that the line $d_0$ passing through $A_0$ and parallel to $A_1 A_2$ serves the same purpose.
	
	\item The line $d$ intersects all triangle $A_i A_{i+1} A_{i+2}$ for all $i \geq 0$, since it intersects $A_i A_{i+1}$ for all $i > 0$.
	
	\item Let $X$ be a point inside the triangle $\triangle A_{n-1} A_n A_{n+1}$ and lying on $d$, then the algorithm would take $n$ iterations even if we incorporate the information of $x$ in Phase 1, and regardless of the pivot rule. The conclusion follows.
\end{enumerate}

What needs to be done is to generalise this construction to higher dimensions. One of the simple ways to do so in $\mathbb{R}^{d+1}$ is to take an arbitrary 2-dimensional subspace $H$ in which we carry out the above construction, and $d-1$ points $C_0, C_1, ..., C_{d-1}$ such that, for instance, the vectors $\overrightarrow{A_0 C_i}$ for $0 \leq i \leq d-1$ forms a basis of the complementary subspace $H^\perp$. Then, one has $P = \conv{\{A_0, A_1, ..., A_{n+1}, C_0, C_1, ..., C_{d-1}\}}$ is a non-degenerate polytope, and the analysis above holds, with the number of iterations remains $n$. An intuitive reason is because if we suppose a line $\ell \subseteq H$ separates two points $A, B \in H$ in $H$, then the affine hyperplane containing $\ell$ and $C_0, C_1, ..., C_{d-1}$ also separates two points $A$ and $B$ in $\mathbb{R}^{d+1}$.

\begin{figure}%
	\centering
	\subfloat[\centering Construction in 2D...]{{\includegraphics[width=0.5\textwidth]{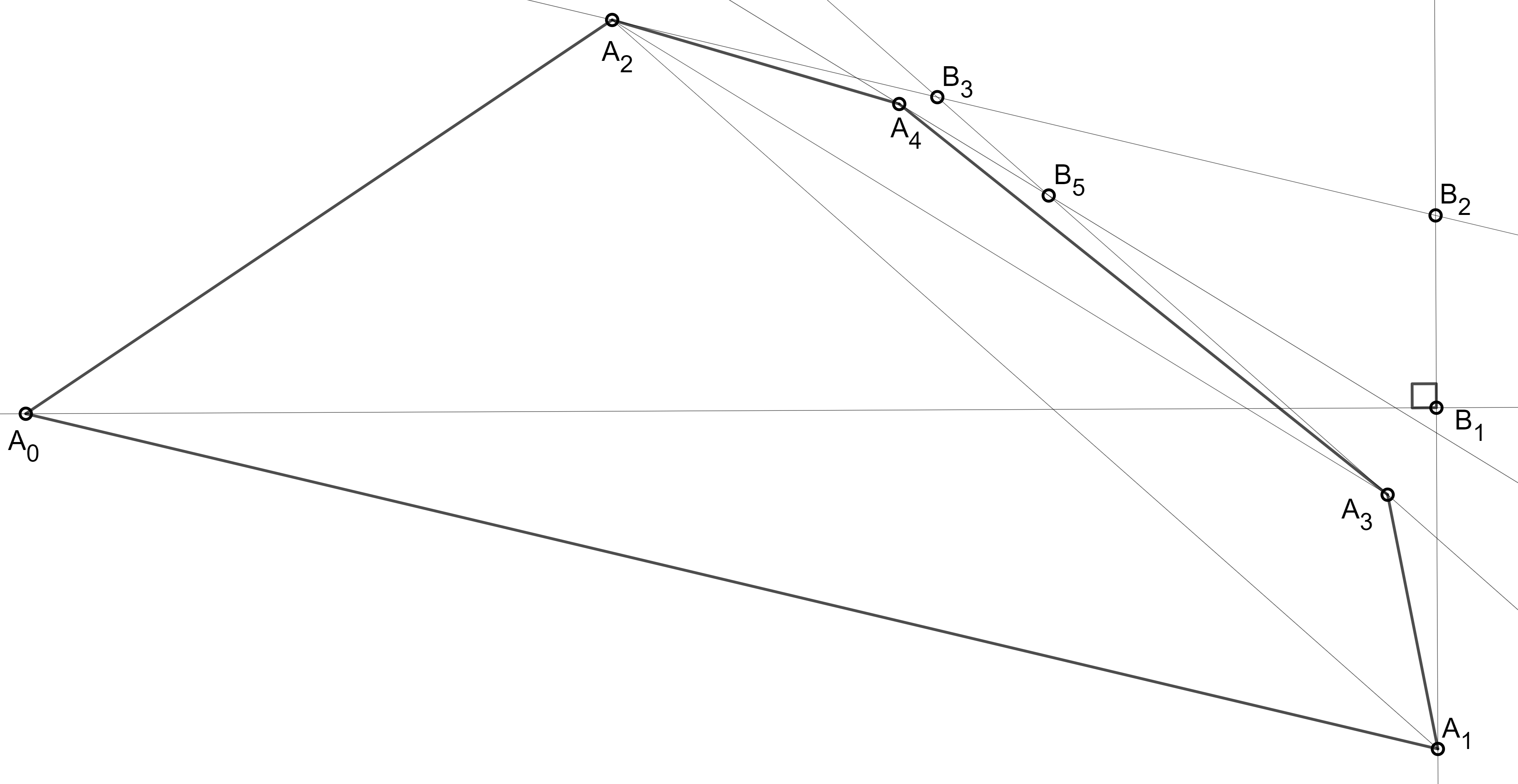} }}%
	\enspace
	\subfloat[\centering which is then lifted to 3D]{{\includegraphics[width=0.45\textwidth]{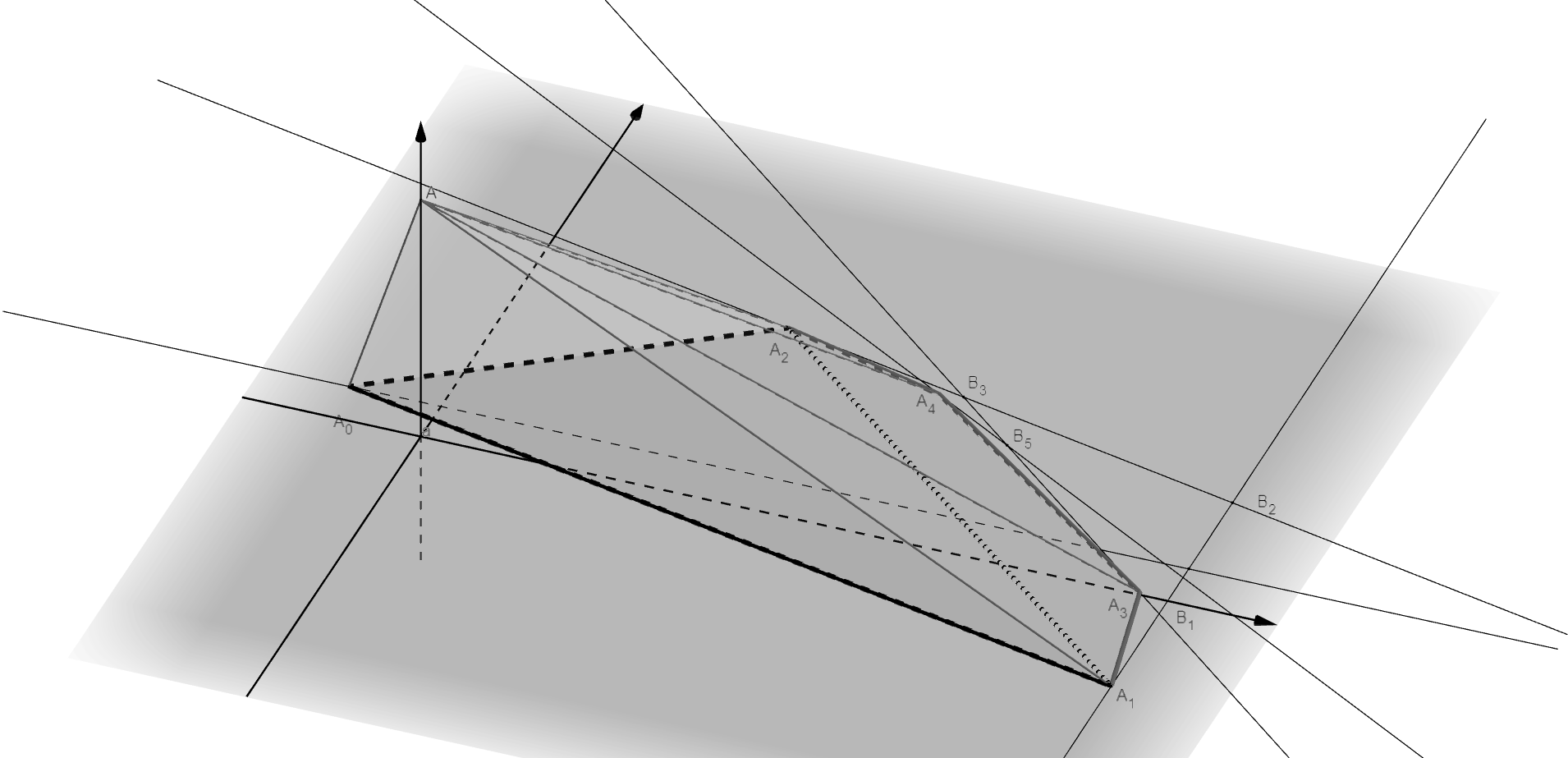} }}%
	\caption{The first few iterations of an construction.}
\end{figure}

\subsection{Conclusion}

In this chapter, we present a mathematical description of a family of methods for the relative position problem, which involves initiating a simplex given by points on the boundary of the concerned polytope $P$ and gradually refining the simplex, until no such improvement can be done, at which point we can conclude about the relative position of $x$.

We also give definition of a pivot rule, and demonstrate the dependence of performance, and even of termination of the scheme on the choice of such a rule. Unfortunately the scope of the study does not allow further examination of the rules.

Finally, assuming the algorithm is executed in real computation model and terminates, we show that the complexity is linear in the number of vertices of $P$, and that this bound is tight, independent from the dimension. Note that in the context of RNA secondary structure, the polytope has potentially vertices exponentially many with respect to the number of features. 
\section{Relative position problem: a theoretical solution}
\label{chapter: ellipsoid method part 2}

As we see in the last chapter, simplex method with an appropriate choice of pivot rule can demonstrate a great performance. Nonetheless, we see that in the worst case, its complexity can be linear with respect with the number of vertices of polytope, which, in the context of RNA polytopes, can be exponential of the sequence's length. On the one hand, in theory, this means we have not solved Relative position problem in polynomial time with respect to the dimension. On the other hand, in practice, a RNA sequence can be as long as 200,000 nts, and even if one limits themselves to the effective range of thermodynamics method, it can still be up to 700 nts. This may prove to be a bottleneck.

In this chapter, inspired by work of Hornus \cite{Hornus2017}, we will develop an algorithm based on the ideas from ellipsoid method in linear programming, which motivates our naming. As we will see, the proof of convergence for both methods are closely related. In particular, it allows us to solve Relative position problem in weakly polynomial time.

\begin{theorem}
	Relative position problem is solvable in weakly polynomial time.
\end{theorem}

Similar to Section \ref{chapter: simplex method part 2}, in this Section, we consider the case where the polytope $P$ is non-degenerate. Otherwise, $P$ has empty interior, so $P = \partial P$, and if $x \in P$, we return $x \in \partial P$. 

\subsection{Method description}

Similar to chapter 4, in what follows we shall assume that $P$ is non-degenerate. Otherwise, we will return $x \in \partial P$ for all $x \in P$.

\subsubsection{Geometric idea}

For $u \in \mathbb{S}^d$, we denote $C(u) = \mathbb{S}^d \cap \{v \mid \langle u, v \rangle = 0\}$,  $C_{\leq 0}(u) = \mathbb{S}^d \cap \{v \mid \langle u, v \rangle \leq 0\}$, and $C_+(u) = \mathbb{S}^d \cap \{v \mid \langle u, v \rangle > 0\}$. Then, we have the two following lemmas.

\begin{lemma}(rephrasing from Hornus's \cite[Lemma 1]{Hornus2017}, 2017)
	Let $x, y \in \partial P$ such that $x \neq y$, $u = \frac{y - x}{\| y - x \|}$ and $S$ be the normal spherical polytope of $x$. Then, one has $S \subseteq C_{\leq 0}(u)$.
\end{lemma}
\begin{proof}
	For any $v \in C_+(u)$, one has $\frac{1}{\| y - x \|} (\langle y, v \rangle - \langle x, v \rangle) = \left\langle \frac{y - x}{\| y - x \|}, v \right\rangle > 0$, so $\langle y, v \rangle > \langle x, v \rangle$ and $v \not\in S$. Therefore, $S \subseteq \mathbb{S}^d \setminus C_+(u) = C_{\leq 0} (u)$.
\end{proof}

\begin{lemma}
	Let $x \in P$, $p \in \mathbb{S}^d$, and $y = \sigma_P(p)$. Assume that $\langle x, p \rangle < \langle y, p \rangle$, and let $u = \frac{y - x}{\| y - x \|}$ which is well-defined since $x \neq y$, then $p \in C_+(u)$.
\end{lemma}
\begin{proof}
	By definition, one has $\langle y, p \rangle = h_P(p) > \langle x, p \rangle$, whence follows
	$ \frac{1}{\| y - x\|} \langle p, y - x \rangle = \langle p, u \rangle > 0$, or $p \in C_+(u)$.
\end{proof}

In simple terms, what these two lemmas say is that assuming $x \in \partial P$ and let $S$ be its normal spherical polytope, then for any vector $p \in \mathbb{S}^d$, if $\langle x, p \rangle < h_P(p)$, then there exists a great circle separating $p$ and $S$, given by $C(u)$ where $u = \frac{\sigma_P(p) - x}{\| \sigma_P(p) - x \|}$. This plays the role of a separating affine hyperplane in our ellipsoid method.

Now to construct the ellipsoid, note that after some $i$ iterations, we bound $S$ by another spherical polytope $T_i$. The idea is then to bound $T_i$ by a spherical ellipsoid $\mathcal{E}_i$ defined as in Chapter 1, and consider the centre $\mathcal{c}_i$ of $\mathcal{E}_i$. If $x$ lies on the boundary of $P$, $S$ will be contained in the region $\mathcal{E}_i \cap C_{\leq 0} (u)$, which can be then bounded by another spherical ellipsoid $\mathcal{E}_{i+1}$. Given that $\mathcal{c}_i \in E_i \cap C_+ (u)$, we can hope that similarly to ellipsoid method in linear programming, we have $\frac{\vol{\mathcal{E}_{i+1}}}{\vol{\mathcal{E}_i}} < C$ for some constant $0 < C < 1$, which leads to an algorithm running in weak polynomial time. As we shall see later, we can associate to each spherical ellipsoid $\mathcal{E} \subseteq \mathbb{S}^d$ an ellipsoid $E \subseteq \mathbb{R}^d$ such that $\vol{\mathcal{E}}$ tends to $0$ as $\vol{E}$ tends to $0$, and we show that $\frac{\vol{E_{i+1}}}{\vol{E_i}} < C$: in fact, we will use in our proof the similar result from ellipsoid method in linear programming.

\subsubsection{Spherical ellipsoid representation}

In $\mathbb{R}^{d+1}$, an ellipsoid $E$ can be represented by a positive-definite matrix $Q$ and a centre $c$ as $E = \{x \mid (x - c)^T Q^{-1} (x-c) \leq 1\}$, which allows closed formula for computing $E_{i+1}$, arguing about the volume, and practical implementation. Unfortunately, spherical ellipsoids in $\mathbb{S}^d$ lack such a representation, which not only makes it difficult to implement the algorithm, but also to compute $\mathcal{E}_{i+1}$ and calculate its volume even for the case $d = 2$. To this, we propose a possible solution, where a spherical ellipsoid $\mathcal{E}$ will be represented by a convex cone $C$ such that $\mathcal{E} = \mathbb{S}^d \cap C$. In particular, $C$ can be represented by a vector $u \in \mathbb{S}^d$ and an ellipsoid $E \subseteq \mathbb{R}^d$ who takes vectors in the canonical basis as its eigenvectors, as follow:
\begin{quote}
	Consider the affine hyperplane $H$ tangent to $\mathbb{S}^d$ at $u$, and endow it with a suitable right-orientated orthonormal basis $u_1, u_2, ..., u_d$. Identify $H$ with $\mathbb{R}^d$ where $u$ is identified with the origin, and for all $i$, $u_i$ is identified with $e_i$ in the canonical basis of $\mathbb{R}^d$. Then, we construct an ellipsoid $E_H$ in $H$ given by $E$ and the identification between $H$ and $\mathbb{R}^d$, and $C$ is given by
	\[
		C = \{\alpha v \mid \alpha \geq 0, v \in E_H\}.
	\]
	Moreover, we can choose a basic such that the basic vectors are the axes of $E_H$, thus one may assume $E$ admits the vectors of the canonical basis of $\mathbb{R}^d$ as its eigenvectors.
\end{quote}

\begin{figure}
	\centering
	\includegraphics[width=\textwidth]{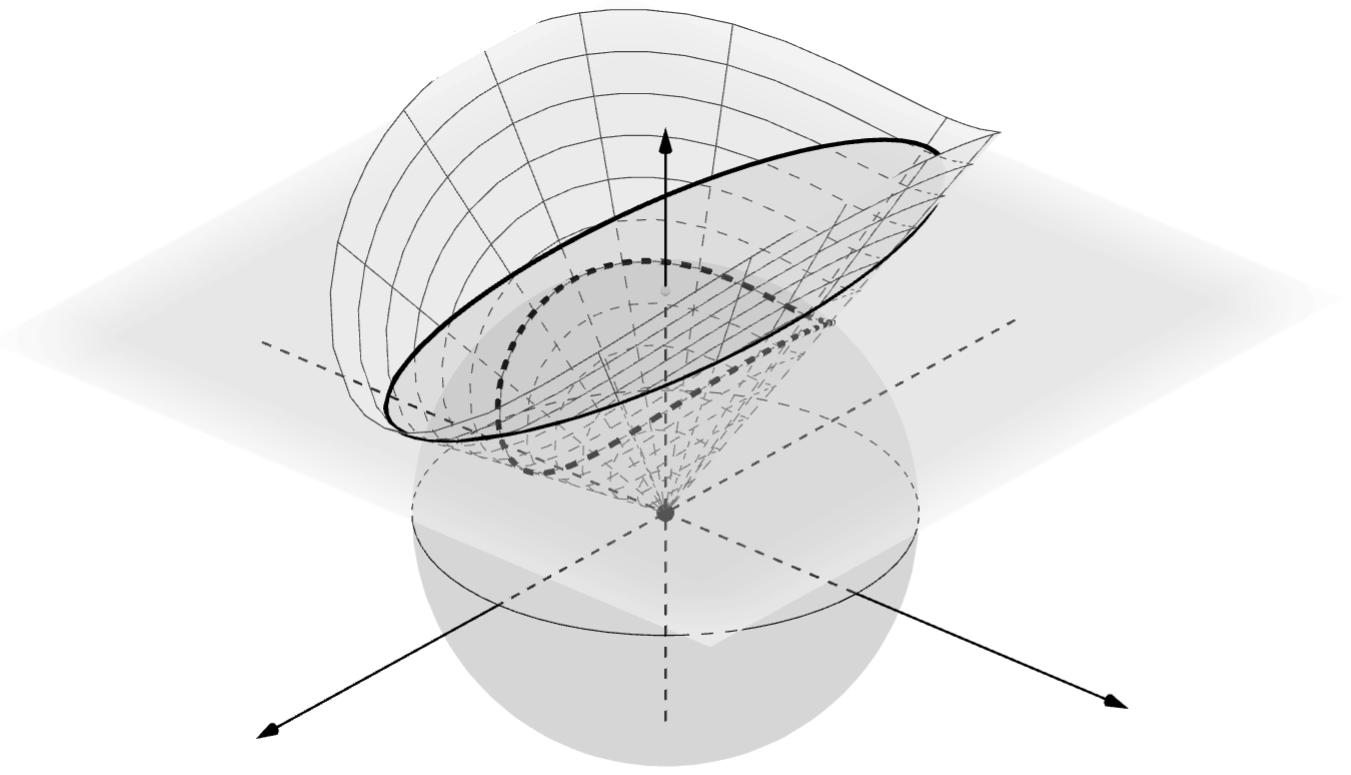}
	\captionsetup{justification=centering, margin=2cm}
	\caption{Spherical ellipsoid and its representation. Solid line and dashed line represent $E_H$ and $\mathcal{E}$, respectively. The plane represents $H$.}
\end{figure}

One can show that such a pair $(u, E)$ is necessarily unique, which gives a representation of $\mathcal{E}$. In what follows, whenever appropriate, to mean a spherical ellipsoid $\mathcal{E}$, we will use such a pair $(u, E)$ of a vector $u \in \mathbb{S}^d$ and an ellipsoid $E \subseteq \mathbb{R}^d$ given by a positive-definite matrix $Q$ as $E = \{x \mid x^T Q^{-1} x \leq 1\}$.

Finally, we need to show that $\vol{\mathcal{E}}$ tends to $0$ as $\vol{E}$ tends to $0$.
\begin{lemma}
	\label{lemma on volume convergence to 0}
	With notations defined as above, as $\vol{E}$ tends to $0$, so does $\vol{\mathcal{E}}$.
\end{lemma}
\begin{proof}
	$\vol{\mathcal{E}}$ tends to $0$ if any only if $\vol{C \cap B_d(0, 1)}$ tends to $0$, where $B_d(0, 1)$ is the unit ball in $\mathbb{R}^{d+1}$. On the other hand, one has $\vol{C \cap B_d(0, 1)} < \vol{C \cap H_{-}}$ where $H_{-} = \{y \mid \langle y, u \rangle \leq 1 \}$ is the halfspace defined by $H$ and containing the origin. And finally, $\vol{C \cap H_{-}} = \frac{\vol{E}}{d+1}$, whence the conclusion follows.
\end{proof}

Similar to ellipsoid method in linear programming, our algorithm has two phases: finding an initial bounding spherical ellipsoid, and refining it.

\subsubsection{Phase 1: Finding an initial bounding spherical ellipsoid}

Similarly to Phase 1 of simplex method presented in the last chapter, we construct a non-degenerate simplex $Q = \conv{\{u_0, u_1, ..., u_{d+1}\}}$ in $d+2$ iterations. Then, we consider the points $\mathcal{u_i} \in \mathcal{S}^d$ corresponding to the facet formed by $u_j$ where $j \neq i$, and since any two facets are adjacent, in particular they cannot be parallel to each other, thus no two points $\mathcal{u_i}$ and $\mathcal{u_j}$ are antipodal. Thus, to these $d+2$ points, we have $d+2$ distinct great circles given by $C(\mathcal{u_0}), C(\mathcal{u_1}), ..., C(\mathcal{u}_{d+1})$, which divide $\mathbb{S}^d$ into spherical polytopes. Now if the algorithm does not terminate after this phase, then it must be the case that the normal spherical polytope $S \in \mathcal{S}(P)$ of $x$, if non-empty, must lies in one of the halfspheres defined by $C(\mathcal{u}_i)$ for all $i = 0, 1, ..., {d+1}$, and therefore amongst the spherical polytopes of dimension $d$ generated by $C(\mathcal{u_0}), C(\mathcal{u_1}), ..., C(\mathcal{u}_{d+1})$, there exists uniquely one containing $S$.

On the other hand, any normal spherical polytope $S' \in \mathcal{S}(Q)$ can be contained in a ball (in $\mathbb{S}^d$) whose radius is strictly less than $\frac{\pi}{2}$. Indeed, suppose the contrary, then there exists some normal spherical polytope $S' \in \mathcal{P}$ and some $p \in \mathbb{S}^d$ such that both $p$ and $-p$ are in $S'$. The polytope $S'$ corresponds to some point $x' \in \partial Q$, and by construction, one has that
\[
Q \subseteq \{y \mid \langle y, p \rangle \leq \langle x', p \rangle \} \cap \{y \mid \langle y, p \rangle \geq \langle x', p \rangle \} = \{y \mid \langle y, p \rangle = \langle x', p \rangle\},
\]
contradicting the non-degeneracy of $Q$.

Therefore, suppose $S$ is contained in some $S' \in \mathcal{S}(Q)$, which in turns is contained in the ball $B_{\mathbb{S}^d}(u, \theta)$ where $0 < \theta < \frac{\pi}{2}$. This is identified with $(u_1, E_1)$ where $E_1 = B(0, r) \subseteq \mathbb{R}^d$ is a ball of radius $r = \tan^{-1} \theta$ necessarily finite.

As for how to construct $u_1$ and $E_1$ explicitly, to each great circle $C$ that contains a facet of $S'$, we denote $H_C$ the hyperplane containing $C$. The collection of such hyperplanes $H_C$ for all $C$ forms a cone $\mathcal{C}$. Intersecting $\mathcal{C}$ with a hyperplane $H$ such that $\mathcal{C} \cap H$ is bounded, then in fact $\mathcal{C} \cap H$ is a polytope in $H$, for which we can construct the minimal bounding ellipsoid $E$. The elliptic cone $\mathcal{C}_E = \{\alpha v \mid \alpha \geq 0, v \in E\}$ then contains $\mathcal{C}$, and we can take $(u_1, E_1)$ to be the representation of $\mathcal{C}_E$.

\subsubsection{Phase 2: Refining spherical ellipsoid}

At $i$\textsuperscript{th} iteration, we have our spherical ellipsoid $\mathcal{E}_i$ represented by the pair $(u_i, E_i)$. As before, we consider the affine hyperplane $H_i$ tangent to $\mathbb{S}^d$ at $u_i$, endow it with a suitable right-oriented orthonormal basis $v_1, v_2, ..., v_d$, and identify $H_i$ with $\mathbb{R}^d$ where $u$ is identified with the origin, and for all $i$, $v_i$ is identified with $e_i$ in the canonical basis of $\mathbb{R}^d$. Also, we construct an ellipsoid $E$ in $H_i$ given by $E_i$ and the identification between $H_i$ and $\mathbb{R}^d$.

Now we test the vector $u_i$ and see if $\langle x, u_i \rangle = h_P(u_i)$. If this is true, then $x \in \partial P$; else, we introduce the separating great circle given by $C(w)$ where $w = \frac{\sigma_P(u_i) - x}{\| \sigma_P(u_i) - x \|}$. Also consider the hyperplane $K$ such that $C(w) = K \cap \mathbb{S}^d$. Let $\ell = K \cap H_i$, we have two cases.
\begin{itemize}
	\item Either $\ell$ does not intersect $E$, at which point we also know that $C(w)$ does not intersect $\mathcal{E}_i$. Given that the centre $\mathcal{c}_i$ lies in $C_+(w)$, it is thus necessarily the case that $\mathcal{E}_i \subseteq C_+(w)$, and we conclude that $S$ is empty, i.e. $x \in \mathring{P}$.
	\item Or $\ell$ intersects $E$. We then construct an ellipsoid of minimal volume $E' \subseteq H$ containing the intersection of $E$ and the half-hyperplane of $H$ defined by $\ell$ and not containing the centre $c_i$ of $E_i$. The exact description can be computed by identifying $\ell$ with a line $\ell'$ in $\mathbb{R}^d$, where if we give a normal vector $l$ for $\ell'$, then $E'$ will correspond to the minimal ellipsoid $E'_i$ containing $E_i \cap \{y | \langle y, l \rangle \leq 0\}$.
\end{itemize}
Note that $E'_i$ does \emph{not} give the representation for $\mathcal{E}_{i+1}$. Whilst one does have that $u_{i+1} = \frac{c'}{\|c'\|}$ where $c'$ is the centre of $E'$, the affine hyperplane $H$ which contains $E'$ is not tangent to $\mathbb{S}^d$ at $u_{i+1}$. To correct this, we then first consider the cone $C' = \{\alpha v | \alpha \geq 0, v \in E'\}$, $H_{i+1}$ be the affine hyperplane tangent to $\mathbb{S}^d$ at $u_{i+1}$, and compute $E_{i+1} = H_{i+1} \cap C'$. Note that let $\mathcal{E}_{i+1} = C' \cap \mathbb{S}^d$, then it must contain the intersection of $\mathcal{E}_i$ and the halfsphere defined by $C(w)$ not containing $\mathcal{c}_i$.

Finally, if $\vol{E_{i+1}} \leq \varepsilon$ for some fixed $\varepsilon > 0$, we terminate and report that $x \not\in P$.

\subsection{Complexity}

With Lemma \ref{lemma on volume convergence to 0}, it is sufficient to exhibit some constant $0 < C < 1$ depending only on $d$ such that $\frac{\vol{E_{i+1}}}{\vol{E_i}} \leq C$, and we demonstrate such a constant by the following theorem.

\begin{theorem}
	\label{theorem on convergence of ellipsoid method}
	With notations defined as above and assume that $d \geq 3$ or $d = 1$, one has that $\frac{\vol{E_{i+1}}}{\vol{E_i}} \leq e^{-\frac{1}{2(d+1)}}$.
\end{theorem}

We have attempted to prove for $d = 2$, but the calculations proved to be cumbersome, and thus we decided to not pursue the case further, and omit it from this thesis for the sake of brevity. We remark that although a priori we do not know if the convergence rate guarantee holds for $d = 2$, this poses no great difficulty, as one can still lift a polytope $P$ from $\mathbb{R}^3$ to $\mathbb{R}^4$ by considering a polytope $Q$ given as follow:
\begin{quote}
	Let $z = (z_1, z_2, z_3)^T \in P$, and consider two points $x_+ = (z_1, z_2, z_3, 1)^T$ and $x_- = -x_+ = (z_1, z_2, z_3, -1)^T$. Identify $\mathbb{R}^3$ with the hyperplane $H = \{x_4 = 0\}$, and let $P' = \{(x_1, x_2, x_3, 0) \mid (x_1, x_2, x_3) \in P\}$. Finally, let $Q = \conv{\{x_+, x_-, P'\}}$.
\end{quote}
We can also write formally
\[
Q = \conv{\left\{\begin{pmatrix}
		z \\ 1
	\end{pmatrix}, \begin{pmatrix}
		z \\ -1
	\end{pmatrix}\right\} \cap \left\{\begin{pmatrix}
		y \\
		0 
	\end{pmatrix} \Big| y \in P\right\}}
\]
Finally, let $x' = \begin{pmatrix}
	x \\
	0
\end{pmatrix}$, and we run the algorithm with $Q$ and $x'$. One can see that $x \in \partial P$ if and only if $x' \in \partial Q$, therefore the algorithm will return correctly.

Thus, for a given $\varepsilon > 0$, the algorithm will terminate after at most $d+2 + 2(d+1)\ln \frac{\vol{E_1}}{\varepsilon}$ iterations, each iteration calls the supporting function $h_P$ and the extremal function $\sigma_P$ exactly once, and all the other linear algebra operations are performed in polynomial time with respect to $d$.

Now recall that for a rational number $x$, its size, denoted by$\langle x \rangle$, is the number of bits needed to represents $x$ (cf. Appendix \ref{appendix: linear programming}). For a vector $x = (x_1, x_1, .., x_{d+1})^T \in \mathbb{Q}^d$, we denote $\langle x \rangle = \langle \max_{i} |x_i| \rangle = \max_{i} \langle |x_i| \rangle$. And we call the size of $P$, denoted $\langle P \rangle$, to be the maximum size of its vertices. 

The final catch is that weakly polynomial time is defined in $d$ and sizes $\langle P \rangle = D$ of coordinates that represent $P$'s vertices. What is left to prove, is that $\langle \vol{E_0} \rangle$ is also bounded by polynomial of $D$. Here, for the sake of brevity, we only present a sketch of the proof.
\begin{enumerate}
	\item Since the coordinates of vertices have size bounded by $D$, the "area" of facets and the volume of $Q$ have size to be polynomial of $D$. This leads to the dihedral angles of $Q$ also having size to be polynomial of $D$.
	\item Moving to normal polytopes, since the smallest distinguishable angle has size of polynomial of $D$, so are the length of $[x, y]_{\mathbb{S}^d}$ for any $x, y \in S$ where $S \in \mathcal{S}(Q)$. It follows that the smallest ball $B_{\mathbb{S}^d}(u, \theta)$ that contains $S$ must also has $\theta$ of size $D$.
	\item Finally, recall that as $x$ tends to $\frac{\pi}{2}$, one has $\tan x = \frac{1}{\tan\left(\frac{\pi}{2} - x\right)} \approx \frac{1}{\frac{\pi}{2} - x}$, so in general, $\tan x$ has size of polynomial with respect to that of $x$. Projecting $S$ on $H$, it follows that $E_1 = B(0, r)$ has its radius $r$ whose size is of polynomial with respect to that of $\tan \theta$.
\end{enumerate}

Therefore, the algorithm complexity is polynomial in $d$ and $\langle P \rangle$, which, together with the method using the relation between Optimisation oracle and Separation oracle as presented at the beginning, solves Relative position problem.

\subsection{Concluding remarks}

We have demonstrated that Relative position problem is solved at least from the theoretical point of view. But, we also remark that in linear programming, ellipsoid method suffers from slow convergence that it is often outperforms by simplex method, despite having the theoretical guarantee. In our problem, this "ellipsoid method" also has the same convergence rate as demonstrated by Theorem \ref{theorem on convergence of ellipsoid method}, thus we can expect the same phenomenon. This method perhaps shall play no more than the role of a theoretical technique, and should we wish to practically solve Relative position problem in higher dimension with a complexity guarantee, another method is needed.

We also admit that we leave out many details in this Section for the sake of brevity. The proof for Theorem \ref{theorem on convergence of ellipsoid method} is shown in Appendix \ref{appendix: proofs}. We have only presented a sketch for the proof that $\langle \vol{E_1} \rangle$ is of polynomial of $\langle P \rangle$. And overall, our choice of spherical ellipsoid representation, whilst being natural geometrically, is unorthodox to work with from symbolic standpoint, for which a more thorough and rigorous description will be helpful.

Moreover, we have not implemented the algorithm, and thus have no remarks on numerical stability or the lack thereof. That being said, the number and the complexity of linear algebraic operations involved pose challenges on how to implement the method with robustness. We leave this issue for future studies.

\section{Future works}
\label{chapter: conclusion}

\subsection{On Relative position problem}

In this concluding remark, we wish to go beyond the scope of the two presented methods, and draw attention to the correspondence between methods in linear programming and ours. That is, the simplex method performs great in practice, but has poor worst-case performance, whereas the ellipsoid method has the theoretical guarantee, but involves unstable numerical operations and has slow convergence. Moreover, the simplex method both in linear programming and in our setting has its performance sensitive to the choice of pivot rule, and the ellipsoid method has similar convergence rate as expected from the fact that the proofs are closely related.

This begs the question if there exists an analogue of the interior-point method for Relative position problem, with potentially good performance in practice and a theoretical convergence guarantee. It is not obvious what would play the role of an interior point in this case, and we leave this direction as open for future research.

In the grand scheme of linear programming, each of the three methods opens up new research directions of their own: study of pivot rules for simplex method, study of cuts for ellipsoid method, and study of barriers for interior-point methods. With the correspondence given above, it is easy to see what directions for the study of Relative position problem. We have demonstrated that MPR algorithm can be outperformed by other pivot rules, but have yet to show any such rules that is guaranteed to terminate. We also have shown a possible representation for spherical ellipsoids, chosen for convenience and the fact that the spherical ellipsoids only represent cones, and it is these cones that are of our interests. Nonetheless, other representations are possible, such as intersections of $\mathbb{S}^d$ with elliptic cylinders, and this may improve numerical stability. We leave this hypothesis to be proven or disproven in the future.

And finally, as an analogue to linear programming, we ask if Relative position problem can be solved in \emph{strongly} polynomial time. Given the dependence upon the relationship between Optimisation oracle and Separation oracle, it is more difficult than that in linear programming, but even assuming $x \in P$, this question is interesting in its own right.

\subsection{Quantifying learnability}

Back to the setting of RNAs, we have shown how to determine if a pair $(q, s)$ is learnable for an energy model $E$, and if so, how to give such a parameter set $p$ and measure its robustness (cf Chapter \ref{chapter: telescoping method part 2}). But, if such a pair is not learnable, it is better to quantify how \emph{not} learnable it is, or equivalently, how close we can get to the experiment results. A way to define a measure of learnability is to compute, considering all possible parameter sets $p$, how close can the predicted energy get to the measurement. In mathematical notations, we wish to see

\[
	\min_{p} [h_{\mathcal{P}(q)}(p) - \langle c(q, s), p \rangle],
\]
or expressed in polytope $\mathcal{P}(q)$, we wish to calculate the distance $d(c(q, s), \partial P)$ from $c(q, s)$ to the boundary $\partial P$.

To see how the notion of learnability may be useful, we return to the basic question stated at the beginning regarding minimum free energy's poor performance. For instance, Workman and Krogh demonstrated that using Turner model, in terms of free energy, observed RNA structures and random sequences with the same dinucleotide distribution are not statistically distinguishable \cite{Workman1999}, suggesting that thermodynamic methods might not be suitable, or at least not sufficient to determine RNA structure. Matthews et al. \cite{Mathews1999a} showed that by considering 750 suboptimal secondary structures with best free energy, the algorithm's accuracy improved from 72.9\% to 97.1\%.

Many reasons are attributed to this phenomenon.
\begin{enumerate}
	\item The parameters are inevitably imprecise, both by limited numerical precisions \cite{Wuchty1999} and lacking in our current understanding of RNA dynamics.
	
	\item Contrary to the dogma that energy parameters are universal, Matthews et al. \cite{Mathews1999a} suggested that they could be sequence-dependent, and Wuchty et al. \cite{Wuchty1999} suggested that there could be unknown physical processes which might changes these parameters, thus making a priorly suboptimal structure more favourable.
	
	\item Minimum free energy methods assume RNA exists and folds to a global minimum, whereas there are no reasons to believe that it is the case: folding pathways have shown to trap RNAs in local minima \cite{Tinoco1990}. Likewise, they fail to account for folding of RNA as it is being transcribed - also known as co-transcriptional folding - and similarly, single-stranded DNAs, which occur during replication, fold as they appear \cite{SantaLucia2004}.
\end{enumerate}

Whilst the last reason is due to the inherent limitation of energy model, the first two can be explained by unsuitable parameters, thus inspiring the study of suboptimal structures. Wuchty et al. described an algorithm to find all suboptimal structures within a given threshold above the minimum free energy, but they also showed that the number of such structures grew rapidly as the length increased: a sequence of 100 nt admits almost 2 millions suboptimal structures within 10 kT from the minimum energy \cite{Wuchty1999}. Thus, it is useful to determine how close, for a given RNA, the observed free energy and a model's prediction can be. Intuitively, the small gaps can be overcome in reality, and suboptimal structures can be accessible: the smaller the gap, the more likely such a case occurs.

Moreover, Wuchty et al. \cite{Wuchty1999} proposed using the density of such suboptimal structures around the global minimum to determine how well-defined the predicted structure is, yet such a measure depends on the choice of parameters. The notion of learnability we consider overcomes this dependence. Unfortunately, none of the two methods we presented in this thesis allow computing $d(x, \partial P)$ for a given point $x \in \mathring{P}$ and a polytope $P$ represented by its supporting function. GJK algorithm allows such a calculation, but it is only an approximation, and as we discussed in Chapter \ref{chapter: introduction part 2}, such a method is not applicable for higher dimensions.

\subsection{Beyond Minimum free energy and RNA secondary structure}

Despite the popularity and variety, Minimum free energy is only one amongst many methods to predict the secondary structures. For instance, to overcome some or all of the issues mentioned in the previous section, another approach, called Maximum expected accuracy, was proposed by McCaskill \cite{McCaskill1990}. The idea was to focus on maximising the probability that a predicted base pair is correct rather than minimising free energy. To do this, McCaskill devised the following partition function
\[
	Z(q, p) \coloneq \sum_{s'} \exp\left(-\frac{1}{RT} \langle c(q, s'), p \rangle \right)
\]
where $R$ and $T$ are the gas constant and the absolute temperature of the environment. This function is efficiently computable using a dynamic programming scheme where RNA is decomposed similarly to that in Minimum free energy. One then has the probability of a given structure $s$ to be
\[
	\mathbb{P}(s \mid q, p) = \frac{1}{Z(q, p)} \exp\left(-\frac{1}{RT} \langle c(q, s), p \rangle\right),
\]
and the probability of a base pair $\{i, j\}$ between $q_i$ and $q_j$ to be $\mathbb{P}_1(\{i, j\} \mid q, p) = \sum_{\{i, j\} \in s} \mathbb{P}(s \mid q, p)$. In similar fashion, the probability that $i$\textsuperscript{th} nucleotide is \emph{not} paired, is $\mathbb{P}_2(i \mid q, p) = \sum_{\forall j, \{i, j\} \not\in s} \mathbb{P}(s \mid q, p) = 1 - \sum_{j} \mathbb{P}_1(\{i, j\} \mid q, p)$. Finally, the expected accuracy to be maximised is defined as

\[
	\mathcal{A}(s) = \gamma \cdot \sum_{\{i, j\} \in s} \mathbb{P}_1(\{i, j\} \mid q, p) + \sum_{\forall j, \{i, j\} \not\in s} \mathbb{P}_2(i \mid q, p), 
\]
where $\gamma$ is some chosen weighted factor.

The choice of $\gamma$ is not trivial, and to our knowledge, has not been well-studied. For instance, it is not clear if $\gamma$ can have an universal value or should be dependent on some features, e.g. species to which RNA belongs, or sequence's length. Moreover, this approach, whilst based directly on a chosen energy model, shows better accuracy and is less prone to inaccuracy parameters \cite{Lu2009}. Though the use of energy model is justified, we remark that our energy models, limited by our understanding of RNA dynamics, are incomplete, and thus even if one choose a priori a model, there are no clear reasons why Minimum free energy and Maximum expected accuracy should share the same parameter sets. In this aspect, the effect of varying parameters for Maximum expected accuracy is not well-understood. 

Along side with structure prediction problem are sequence alignment and phylogenetic tree construction, of which we shall not go into details. The essential question posed in both problems is how to compare RNA/DNA sequences, and a possible approach is comparing RNA polytopes arose from an energy model chosen a priori, e.g. by Hausdorff distance. Unlike any single values arose from linear cost model, the polytopes encode much more information, and one can argue that closely related sequences should share similar features, so their polytopes will resemble each others. Unfortunately, this path has not been intractable since we cannot construct the whole polytopes efficiently. 

On the one hand, Hausdorff distance $\delta(P, Q)$ between two polytopes $P$ and $Q$ has a formulation in terms of supporting functions, as $\delta(P, Q) = \max_{u \in \mathbb{S}^d} |h_P(u) - h_Q(u)|$. On the other hand, the function $f(u) = |h_P(u) - h_Q(u)|$ is not convex in general, thus maximising $f$ needs not necessarily be easy. Nonetheless, this link suggests there can be a way to compute or approximate $\delta(P, Q)$ efficiently using only supporting functions. We leave this question for future studies.

\newpage
\bibliographystyle{plain}
\bibliography{automatedrefs, manualrefs}

\newpage
\appendix

\section{Review on methods for linear programming}
\label{appendix: linear programming}

In this appendix, we recall definition of a linear programming problem and give the geometric idea of the simplex method in the context of linear programming, expressed in the language of polyhedra (that is, an intersection of some halfspaces, which needs not be bounded as opposed to a polytope). Note that this is not to give a detailed description of methods in linear programming, and readers are advised to consult other introductory texts, e.g. Matoušek and Gärtner's \emph{Understanding and Using Linear Programming} \cite{Matousek2007}.

A linear programming problem is defined as a problem of the following form
\[
\begin{split}
	\text{maximise   } & c^Tx \\
	\text{subject to } & Ax \leq b \\
	\text{and        } & x \geq 0
\end{split}
\]
where $b \in \mathbb{R}^m$, $c \in \mathbb{R}^n$ and $A \in \mathbb{R}^{m\times n}$. Geometrically, if we define $P = \{x | Ax \leq b, x \geq 0\}$ be a polyhedron, then this amounts to finding $h_P(c)$. It can thus be shown that the maximum, if attainable, will be attained by a vertex of $P$. In particular, if $P$ is bounded, i.e. a polytope, then the maximum is always attained.

\subsection*{Simplex method}

The key idea of the simplex method is thus to obtain a vertex $y \in \vertex{P}$, and one observes that
\begin{itemize}
	\item either $y$ is an optimal solution, i.e. $\langle c, y \rangle = h_P(c)$, 
	\item or there exists another edge incident to $y$ along which the objective function increases strictly. In particular, if this edge connects $y$ with a neighbour vertex $x$, then $\langle c, y \rangle < \langle c, x \rangle$.
\end{itemize}
If this edge extends indefinitely, as possible in case where $P$ is a polyhedron, then the objective function admits no maximum. Else, one can replace $y$ by $x$ if $y$ is not optimal, and repeat the process until no such $x$ can be found, at which point one arrives at an optimal solution.

The remaining issue is to decide which $x$ to choose in case $y$ has multiple possible neighbour vertices $x$ that increase the objective function, i.e. $\langle c, y \rangle < \langle c, x \rangle$. Such a procedure to determine the neighbour vertex is called a pivot rule, and the performance of the simplex method is sensitive to the choice of pivot rule. Unfortunately, although there are rules proven to be terminating, all have shown to exhibit exponential runtime in specifically constructed examples.

\subsection*{Ellipsoid method}

The reason behind exponential complexity in the worst case of the simplex method is that it relies on vertices of the polyhedron $P$ (or, in case $P$ bounded, a polytope). Whilst the number of facets is bounded by that of constraints, the number of vertices does not enjoy the same bound, and in fact can be grow exponentially fast (for instance, Klee-Minty cube, a family of deformed cubes where one in dimension $d$ has $2d$ facets but $2^d$ vertices).

Ellipsoid method instead seeks not to rely on vertices of $P$, but to bound the region of possible solutions by a convex set, and each iteration involves dividing the set into two halves, of which we will eliminate one and bound the other by a similar convex set. One can draw analogy with binary search in this regard. Note that there are many variants of ellipsoid method, such as by using deep cut, surrogate cuts, or shallow cuts. This is not of our interest, and we shall only present the basic ellipsoid method.

The algorithm has two phases:
\begin{enumerate}
	\item One first finds an ellipsoid $E_0$ containing $P$.
	\item At $i$\textsuperscript{th} iteration, we consider the centre $c_i$ of the ellipsoid $E_i$, and ask if $c_i$ belongs to $P$ by checking all the constraints in the system $Ax \leq b$. We have two cases:
	\begin{itemize}
		\item If $c_i \in P$, then we may discard the region $\{x \mid c^T x < c_i\}$ since we know the optimal solution $x^*$ must have $c^T x^* \geq c^T c_i$. In particular, let $E_{i+1}$ be the ellipsoid of smallest volume containing $E_i \cap \{x \mid c^T x \geq c^T c_i\}$.
		
		\item Similarly, if $c_i \not\in P$, then we may discard the region $\{x \mid c^T x > c_i\}$ since we know the optimal solution $x^*$ must have $c^T x^* \leq c^T c_i$. In particular, let $E_{i+1}$ be the ellipsoid of smallest volume containing $E_i \cap \{x \mid c^T x \leq c^T c_i\}$.
	\end{itemize}
\end{enumerate}
If at some point, $E_i = \emptyset$ then we conclude that there is no solution. Otherwise, the second phase repeats until $\vol{E_{i+1}}$ is smaller than some fixed $\varepsilon$, at which point we return $c_i$ as a solution. Denote $\vol{E}$ the volume of an ellipsoid $E$, then one has $\frac{\vol{E_{i+1}}}{\vol{E_i}} \leq e^{-\frac{1}{2(d+2)}}$, thus the algorithm terminates in $2(d+2)\ln \frac{\vol{E_0}}{\varepsilon}$ iterations.

Now, it is important to note that this scheme is \emph{weakly} polynomial-time, i.e. in bit model, and not strongly polynomial-time. To state more formally, for $i \in \mathbb{Z}$, to encode $i$, one may use $\langle i \rangle = \lceil \log_2(|i| + 1) \rceil + 1$ bits. Then, for $r = \frac{p}{q} \in \mathbb{Q}$, one may use $\langle r \rangle = \langle p \rangle + \langle q \rangle$ bits to encode $r$. Now suppose $A$, $b$, and $c$ have rational coefficients, let $B$ be the largest absolute value of the coefficients in $A$, $b$, and $c$, then the runtime of ellipsoid method in bit model is polynomial in $n$, $m$, and $\langle B \rangle$.

On the other hand, for any integer $N$, one can construct a linear programming problem with only $2$ variables and $2$ constraints such that the ellipsoid method runs at least $N$ iterations (note that such a program necessarily has $\langle B \rangle$ tends to infinity as $N$ gets larger). In particular, one cannot remove dependency on $\langle B \rangle$.

\subsection*{Separation problem and Optimisation problem}

It is worthy to note the corresponding between Optimisation problem, and Separation problem, which is given as follow:
\begin{quote}
	Given a polytope $P \subseteq \mathbb{R}^{d+1}$ and $x \in \mathbb{R}^{d+1}$, decide if $x \in P$, and if not, then return a vector $c$ such that $P \subseteq  \{y \mid c^T y < c^T x\}$.
\end{quote}
Such a procedure is called a \emph{separation oracle}. Such a vector $c$, if exist, represents an affine hyperplane $H$ that separates $P$ and $x$, meaning $H$ divides $\mathbb{R}^{d+1}$ into two closed halfspaces, each of which contains either $x$ or $P$, and $H$ does not simultaneously contain $x$ and intersect $P$. If $x \not\in P$, then by Hyperplane Separation Theorem, such a vector $c$ must exist.

On the one hand, the ellipsoid method uses calls to a separation oracle when it decides if $c_i \in P$, and if one ignores the complexity of oracle, then in bit model, ellipsoid method runs in polynomial time with respect to $n$ and $B$ (independent from $m$). Thus, the Optimisation problem can be reduced to Separation problem in polynomial time.

On the other hand, suppose we have an oracle that solves linear programming problems, which we call an optimisation oracle for consistency's sake, that for a given $c \in \mathbb{R}^n$, returns an optimal solution to $\max_{x \in P} c^T x$. Then we can also reduce Separation problem to Optimisation in polynomial time, via polar dual polyhedron $P^\circ = \{p \mid \langle x, p \rangle \leq 1 \}$.

Indeed, given $A \in \mathbb{R}^{m \times n}$ and $b \in \mathbb{R}^m$, let $P$ be the polyhedron given by $P = \{x \mid Ax \leq b\}$, and assume $P$ is non-degenerate.

For a given $c \in \mathbb{R}^n$, an optimisation oracle gives a solution $x \in P$ that maximises $c^T x$. If $c^T x \leq 1$, then $c \in P^\circ$; else, one has $P^\circ \subseteq \{d \mid d^T x < c^T x\}$. Thus, it effectively acts as a separation oracle for $P^\circ$. If the hypotheses to apply Ellipsoid method for $P^\circ$ are satisfied, then one may apply this separation oracle to optimise over $P^\circ$.

Now recall that if $P$ is a compact convex set containing the origin in its interior, then so is $P^\circ$, and one has $(P^\circ)^\circ = P$. Thus, by possibly a translation, one may assume that $P$ contains the origin in its interior, and apply the above optimisation oracle over $P^\circ$ as a separation oracle over $P$. The following diagram illustrates these relations.

\[
\begin{tikzcd}[row sep=5em, column sep=7em]
	\text{Separation oracle over } P \arrow{r}{\text{Ellipsoid method}} & \text{Optimisation oracle over } P \arrow{d}{\text{Polarity}} \\%
	\text{Optimisation oracle over } P^\circ \arrow{u}{\text{Polarity, } (P^\circ)^\circ = P} & \arrow{l}{\text{Ellipsoid method}} \text{Separation oracle over } P^\circ
\end{tikzcd}
\]

In summary, a separation oracle gives rise to an optimisation oracle, and vice versa. Note that in the context of relative position problem, we assume access to $h_P$ - an optimisation oracle over $P$, thus we can effectively decide if $x \in P$.
\section{Proofs}
\label{appendix: proofs}

\subsection*{Proof of Theorem \ref{theorem on decreasing balls}}

First, we have a technical lemma.

\begin{lemma}
	\label{lemma on intersection}
	Let $P \subseteq \mathbb{S}^d$ be a polytope, $x \in \mathring{P}$, $y \not\in P$. Then $[x, y]_{\mathbb{S}^d}$ intersects $\partial P$.
\end{lemma}
\begin{proof}
	Consider the map $f : [0, 1] \mapsto [x, y]_{\mathbb{S}^d}$, where $f(t) = \frac{tx + (1-t)y}{\|tx + (1-t)y\|}$. We denote $T = \{t | f(t) \in P\}$, and $t^* = \sup T$. A priori we know that $0 < t^* < 1$. Let $z = f(t^*)$, if $z \not\in P$, then there exists $\theta > 0$ such that $B_{\mathbb{S}^d}(z, \theta) \cap P = \emptyset$ and $y \not\in B_{\mathbb{S}^d}(z, \theta)$. Let $\varepsilon = \frac{\theta}{d_{\mathbb{S}^d}(x, y)}$ which is well-defined since $x \neq y$, one has $[\theta - \varepsilon, \theta + \varepsilon] \cap T = \emptyset$, a contradiction. Hence $z \in P$.
	
	Now in similar fashion, if $z \in \mathring{P}$, then there exists $\theta > 0$ such that $B_{\mathbb{S}^d}(z, \theta) \subsetneq P$ and $x \not\in B_{\mathbb{S}^d}(z, \theta)$. Let $\varepsilon = \frac{\theta}{d_{\mathbb{S}^d}(x, y)}$, one has $[\theta - \varepsilon, \theta + \varepsilon] \subsetneq T$, a contradiction. Thus $z \in \partial P$, as desired.
\end{proof}

Now suppose the contrary, there exist two different facets of $P$ containing $p^*$. Since $p \neq p^*$, the segment $[p, p^*]$ cannot be perpendicular to both of the facets, hence there must exist a facet $F$ of $P$ containing $p^*$ such that $[p, p^*]$ is not perpendicular to $F$. Now consider the projection $p'$ of $p$ onto $F$. If $p' \in F$, then since 
\[
d_{\mathbb{S}^d}(p, \partial P) = d_{\mathbb{S}^d}(p, p^*) >  d_{\mathbb{S}^d}(p, p') \geq d_{\mathbb{S}^d}(p, \partial P),
\]
we have a contradiction. Otherwise, by Lemma \ref{lemma on intersection}, $[p, p']$ intersects $\partial P$ at a point $p''$, which necessarily differs from $p'$ and thus implies that 
\[
d_{\mathbb{S}^d}(p, \partial P) = d_{\mathbb{S}^d}(p, p^*) > d_{\mathbb{S}^d}(p, p') \geq d_{\mathbb{S}^d}(p, p'') > d_{\mathbb{S}^d}(p, \partial P),
\]
a contradiction, as desired.

\subsection*{Proof of Theorem \ref{theorem on convergence of ellipsoid method}}
A similar result from ellipsoid method in linear programming states that $\frac{\vol{E'}}{\vol{E_i}} \leq e^{-\frac{1}{2(d+1)}}$, and we need to show that $\vol{E_{i+1}} \leq \vol{E'}$, or equivalently, the correction does not increase the volume of the ellipsoid. First, we decompose the correction into two steps:
\begin{enumerate}
	\item We "rotate" the hyperplane $H$ into another hyperplane $H'$ passing through $c'$ and admitting $c'$ as its normal vector, i.e. $H'$ is parallel with $H_{i+1}$. Let $E'' = C' \cap H'$ be the ellipsoid $E'$ but "rotated" to hyperplane $H'$.
	\item We move the hyperplane $H'$ toward the origin until it touches $\mathbb{S}^d$, i.e. when it coincides with $H_{i+1}$.
\end{enumerate}
Let $e_1, e_2, ..., e_d$ be the canonical basis of $\mathbb{R}^d$. By possibly a rotation, assume without loss of generality that $u_i = (0, ..., 0, 1)^T$, so $H = \left\{\begin{pmatrix}
	z \\
	0
\end{pmatrix} | z \in \mathbb{R}^d\right\}$. Write $c' = \begin{pmatrix}
	f \\
	1
\end{pmatrix}
$ where $f = (f_1, f_2, ..., f_d)^T \in \mathbb{R}^d$. Moreover, let $P$, $P'$, and $P_{i+1}$ be the positive-definite matrices corresponding to $E'$ in $H$, $E''$ in $H'$, and $E_{i+1}$ in $H_{i+1}$, respectively. With possibly a rotation, we can assume without loss of generality that $P$ admits the canonical basis of $\mathbb{R}^d$ as its eigenvectors. Let $\lambda_1, \lambda_2, ..., \lambda_d$ be the eigenvalues of $P$, and for $j = 1, 2, ..., d$, let $y_j = \begin{pmatrix}
	\sgn{f_j} \lambda_j e_j + f\\
	1
\end{pmatrix}$, where $\sgn{x} = \begin{cases}
	1 & \text{ if } x \geq 0 \\
	-1 & \text{ if } x < 0
\end{cases}$.

We describe the projection from $H$ to $H'$ where the point $y = \begin{pmatrix}
	z \\
	1
\end{pmatrix} \in H$ is mapped to the point $\mu(y) y \in H'$, where $z \in \mathbb{R}^d$. The condition $\mu(y) y \in H'$ is equivalently to 
\[
f^T f + 1 = \| c' \|^2 = \left\langle \mu(y) y, c' \right\rangle = \mu(y) (z^T f + 1) \Rightarrow \mu(y) = \frac{f^T f + 1}{z^T f + 1}.
\]
In particular, $\mu(y_j) = \mu_j = \frac{f^T f + 1}{\lambda_j \sgn{f_j} f_j + f^T f + 1} = \frac{f^T f + 1}{\lambda_j |f_j| + f^T f + 1}$

One key insight, whose proof shall be omitted for the sake of brevity, is under this projection, axes of $E'$ are mapped to those of $E'$. Whilst we cannot at the moment explicitly endow $H'$ with a basis, one can still compute the eigenvalues of $P'$ by $\| \mu(y_j) y_j - c' \|$. And finally by multiplying with a scaling factor $\frac{1}{\|c'\|}$ , we determine the eigenvalues $\delta_j$ of $P_{i+1}$. Let $F = f^T f + 1$, calculation shows
\[
\begin{split}
	\| \mu(y_j) y_j - c' \|^2
	& = \left\| \mu_j \begin{pmatrix}
		\sgn{f_j} \lambda_j e_j + f\\
		1
	\end{pmatrix} - \begin{pmatrix}
		f \\
		1
	\end{pmatrix}\right\|^2 = 
	\left\| \begin{pmatrix}
		\sgn{f_j} \mu_j \lambda_j e_j + (\mu_j - 1) f\\
		\mu_j - 1
	\end{pmatrix}\right\|^2 \\
	& = (\mu_j - 1)^2 + (\sgn{f_j} \mu_j \lambda_j - (\mu_j - 1) f_j)^2 + (\mu_j - 1)^2 \sum_{j \neq i} f^2_j \\
	& = (\mu_j - 1)^2 + \mu^2_j \lambda^2_j - 2\sgn{f_j}\mu_j (\mu_j - 1) \lambda_j f_j + (\mu_j - 1)^2 f^2_j + (\mu_j - 1)^2 \sum_{k \neq j} f^2_k \\
	& = \mu^2_j \lambda^2_j - 2\mu_j (\mu_j - 1) \lambda_j |f_j| + (\mu_j - 1)^2 F \\
	& = \left(\frac{F}{\lambda_j |f_j| + F}\right)^2 \lambda^2_j - 2 \frac{F}{\lambda_j |f_j| + F} \left(\frac{F}{\lambda_j |f_j| + F} - 1\right) \lambda_j |f_j| \\
	& + \left(\frac{F}{\lambda_j |f_j| + F} - 1\right)^2 F \\
	& = \left(\frac{F}{\lambda_j |f_j| + F}\right)^2 \lambda^2_j + 2 \frac{F\lambda^2_j f^2_j}{(\lambda_j |f_j| + F)^2} + \frac{\lambda^2_j f^2_j}{(\lambda_j |f_j| + F)^2} F \\
	& = \left(\frac{F}{\lambda_j |f_j| + F}\right)^2 \lambda^2_j + 3 \frac{F\lambda^2_j f^2_j}{(\lambda_j |f_j| + F)^2}\\
\end{split}
\]
\[
\begin{split}
	\Rightarrow \delta^2_j
	& = \frac{1}{F} \| \mu(y_j) y_j - c' \|^2 = \frac{3f^2_j\lambda^2_j}{\lambda_j |f_j| + F} + \frac{F\lambda^2_j}{(\lambda_j |f_j| + F)^2} = \frac{3f^2_j + F}{(\lambda_j |f_j| + F)^2} \lambda^2_j\\
\end{split}
\]
By Cauchy-Schwartz inequality, one obtains
\[
\begin{split}
	\prod_{j = 1}^{d} \frac{3f^2_j + F}{(\lambda_j |f_j| + F)^2}
	& \leq \prod_{j = 1}^{d} \frac{3f^2_j + F}{F^2} = \prod_{j = 1}^{d} \left(\frac{1}{F} + \frac{3f^2_j}{F^2}\right) \\
	& \leq \left[\frac{1}{d}\sum_{j=1}^{d} \left(\frac{1}{F} + \frac{3f^2_j}{F^2} \right)\right]^d = \left[\frac{1}{F} + \frac{3}{dF}- \frac{3}{dF^2}\right]^d \leq 1
\end{split}
\]
where the last inequality holds for $d \geq 3$, as
\[
	1 - \frac{1}{F} - \frac{3}{dF} + \frac{3}{dF^2} = \left(1 - \frac{1}{F}\right)\left(1 - \frac{3}{dF}\right) \geq 0 \Rightarrow \frac{1}{F} + \frac{3}{dF} - \frac{3}{dF^2} \leq 1.
\]
As for the case $d = 1$, it is clear that one has $\lambda_1 = 2|f_1|$, so in fact we have a stronger inequality,
\[
	\frac{3f^2_1 + F}{(\lambda_1 |f_1| + F)^2} = \frac{4f^2_1 + 1}{(3f^2_1 + 1)^2} \leq 1.
\]

Finally, we conclude by remarking that
\[
	\vol{E_{i+1}} = \vol{B_d(0, 1)}\left(\prod_{j = 1}^d \delta_j\right)^{\frac{1}{2}} \leq \vol{B_d(0, 1)}\left(\prod_{j = 1}^d \lambda_j\right)^{\frac{1}{2}} = \vol{E'}, 
\]
where $B_d(0, 1)$ denotes the unit ball in $\mathbb{R}^d$, as desired.

\end{document}